\documentclass[10pt,journal,compsoc]{IEEEtran}\usepackage{floatrow}
\usepackage[table]{xcolor}
\usepackage{times}
\usepackage{booktabs}       
\usepackage{mathrsfs}

\usepackage[utf8]{inputenc} 
\usepackage[T1]{fontenc}    
\usepackage{hyperref}       
\usepackage{url}            
\usepackage{booktabs}       
\usepackage{amsfonts}       
\usepackage{nicefrac}       
\usepackage{microtype}      
\usepackage[utf8]{inputenc}
\usepackage{amsmath}
\usepackage{algorithm}
\usepackage{algorithmic}

\usepackage{fixltx2e}

\usepackage{graphicx}

\usepackage{subcaption}              

\usepackage[utf8]{inputenc}

\usepackage{amsfonts}
\usepackage{amsmath}
\usepackage{mathtools}

\usepackage{amssymb}

\usepackage{bm}

\usepackage{color,verbatim}
\usepackage{multirow}
\usepackage{accents}

\usepackage{theoremref}

\usepackage{flushend}

\usepackage{url}

\newcounter{rulecounter}
\newcommand{\resetrule}{ \setcounter{rulecounter}{0}}
\resetrule

\newenvironment{myitemize}%
  {\begin{itemize}%
    \setlength{\itemsep}{0pt}%
    \setlength{\parskip}{0pt}}%
  {\end{itemize}}
\newsavebox{\selvestebox}
\newenvironment{colbox}[1]
  {\newcommand\colboxcolor{#1}%
   \begin{lrbox}{\selvestebox}%
   \begin{minipage}{\dimexpr\columnwidth-2\fboxsep\relax}}
  {\end{minipage}\end{lrbox}%
   \begin{center}
   \colorbox{\colboxcolor}{\usebox{\selvestebox}}
   \end{center}}

\definecolor{orange}{rgb}{1,0.8,0}
\definecolor{gray}{rgb}{.9,0.9,0.9}
\definecolor{darkgray}{rgb}{.3,0.3,0.3}
\definecolor{darkblue}{rgb}{.1,0.0,0.3}
\definecolor{lightblue}{rgb}{0.7,0.7,1}
\definecolor{lightred}{rgb}{1,0.7,.7}
\definecolor{purple}{RGB}{204,153,255}
\definecolor{lightgray}{rgb}{.95,0.95,0.95}
\definecolor{lightgreen}{rgb}{0.3,0.5,0.3}
\definecolor{darkgreen}{rgb}{0.05,0.3,0.05}

\newcommand{\diag}[1]{\mathop{\rm diag}{#1}}

 \newcommand{\define}{:=}

\newtheorem{myproposition}{Proposition}
\newtheorem{myremark}{Remark}
\newtheorem{myproblemstatement}{Problem Statement}
\newtheorem{mylemma}{Lemma}
\newtheorem{mytheorem}{Theorem}

\newtheorem{mycorollary}{Corollary}

\usepackage{pgfplots}
\pgfplotsset{compat=newest}
\usetikzlibrary{plotmarks}
\usetikzlibrary{arrows.meta}
\usetikzlibrary{patterns}
\usepgfplotslibrary{patchplots}
\usepackage{grffile}
\usetikzlibrary{pgfplots.groupplots}
\pgfplotsset{plot coordinates/math parser=false}
\usetikzlibrary{pgfplots.units}
\pgfplotsset{plot coordinates/math parser=false}
\newlength\mywidth
\newlength\myheight
\newlength\mywidths
\newlength\mywidthbb
\newlength\myheightbb
\newlength\mywidthb
\newlength\myheightb
\newlength\myheights
\definecolor{GraphSAC}{RGB}{228,26,28}%
\definecolor{Amen}{RGB}{55,126,184}%
\definecolor{Gae}{RGB}{77,175,74}
\definecolor{Radar}{RGB}{152,78,163}
\definecolor{Cut}{RGB}{255,127,0}%
\definecolor{Flake}{RGB}{153,153,153}%
\definecolor{Conductance}{RGB}{166,86,40}
\definecolor{Avg}{RGB}{247,129,191}

\definecolor{diff-nonlin-invariant}{rgb}{1,0.04700,0.04100}%
\definecolor{diff-nonlin-raw}{rgb}{0.05000,0.32500,0.89800}%
\definecolor{custom-nonlin-invariant}{rgb}{0.92900,0.69400,0.12500}%
\definecolor{custom-nonlin-raw}{rgb}{0.49400,0.18400,0.55600}%
\definecolor{custom-lin-invariant}{rgb}{0.46600,0.67400,0.18800}%
\definecolor{custom-lin-raw}{rgb}{0.30100,0.74500,0.93300}%
\definecolor{monicCubic-nonlin-invariant}{rgb}{1,1,0.08400}
\definecolor{monicCubic-nonlin-raw}{rgb}{0,1,0.08400}
\definecolor{pmonicCubic-nonlin-raw}{rgb}{1,0,1}
\definecolor{pmonicCubic-nonlin-invariant}{rgb}{1,0,0.5}
\definecolor{pointnet}{rgb}{0,0,0.08400}
\definecolor{pointnetplus}{rgb}{0.3,0.3,0.300}

\definecolor{voxnet}{rgb}{0.9,0.9,0.300}
\definecolor{3dshapenets}{rgb}{0.9,0,0.900}

\definecolor{edge}{cmyk}{1,0.74,0,0.77}
\definecolor{node}{cmyk}{0.99,0.66,0,0.57} 
\definecolor{graph}{cmyk}{0.38,0.17,1,0.17} 
\definecolor{binding}{cmyk}{0.04,1,0,0.31}
\definecolor{summarize}{cmyk}{0,1,0.26,0.12} 
\definecolor{DiffPool}{rgb}{0.5,0.5,1}%
\definecolor{DiffScater}{rgb}{1,0.04700,0.04100}%
\definecolor{MonicCubic}{rgb}{0.05000,0.32500,0.89800}%
\definecolor{TightHann}{rgb}{0.92900,0.69400,0.12500}%
\definecolor{pDiffScater}{rgb}{1,0.04700,1}%
\definecolor{pMonicCubic}{rgb}{0.05000,0.800,0.89800}%
\definecolor{pTightHann}{rgb}{0.92900,0.9400,0.12500}%

\definecolor{CGTF}{rgb}{1,0.04700,0.04100}%
\definecolor{SNMF}{rgb}{0.05000,0.32500,0.89800}%
\definecolor{CGTF1}{rgb}{0.92900,0.69400,0.12500}%
\definecolor{SNMF1}{rgb}{0.49400,0.18400,0.55600}%
\definecolor{CMTF}{rgb}{0.46600,0.67400,0.18800}%
\definecolor{PARAFAC}{rgb}{0.30100,0.74500,0.93300}%
\definecolor{NTF}{rgb}{1,1,0.08400}%

\newcommand{\mylinewidth}{1.5}
\newcommand{\markwidth}{1.5}
\newcommand{\scaletrees}{0.6}
\newcommand{\legendfontsize}{\normalsize}

\def \thisplotscale {0.75}
\pgfplotsset{
compat=1.11,
legend image code/.code={
\draw[mark repeat=2,mark phase=2]
plot coordinates {
(0cm,0cm)
(0.15cm,0cm)        
(0.3cm,0cm)         
};%
}
}
\def \unit {\thisplotscale cm}

\tikzstyle {node} = 
	[circle, 
	 draw,inner sep=0pt,
	 minimum width = 0.1*\unit,
	 minimum height = 0.1*\unit,
	 anchor = center]
 \tikzstyle {mnode} = 
	[circle, 
	 draw,inner sep=0pt,
	 minimum width = 0.03*\unit,
	 minimum height = 0.03*\unit,
	 anchor = center]    
\def\colorNode{node}
\def\colorEdge{edge}
\def\colorGraphFilter{graph}
\def\colorSummarizer{summarize}
\def\colorBinding{binding}

\def\dist{1.5}

\def\xlbl{0.1*\dist}
\def\ylbl{0.1*\dist}

\def\myLineWidth{0.5}

\def\myArrowWidth{1}
\def\mmyArrowWidth{0.6}
\def\angleSummarizer{160}
\def\radiusSummarizer{1.5}
\def\shortsize{10}

\def\halfImageWidth{8*\unit}
\def\distBetweenLayers{0.2*\halfImageWidth}
\def\distNodesLayerOne{0.6*\halfImageWidth}
\def\distNodesLayerTwo{0.3*\distNodesLayerOne}
\def\distNodesLayerThree{0.8*\distNodesLayerTwo}
\def\distNodesLayerFour{0.8*\distNodesLayerThree}

\def\mdistBetweenLayers{0.25*\halfImageWidth}
\def\mdistNodesLayerOne{0.5*\halfImageWidth}
\def\mdistNodesLayerTwo{0.3*\distNodesLayerOne}
\def\mdistNodesLayerThree{0.3*\distNodesLayerTwo}
\def\mdistNodesLayerFour{0.1*\distNodesLayerThree}
\newcommand{\featvec}{\mathbf{x}}

\newcommand{\nonlinearity}[1]{\sigma(#1)}

\newcommand{\shiftentry}{S}
\newcommand{\shiftmat}{\mathbf{\shiftentry}}

\newcommand{\eigenval}{\lambda}

\newcommand{\framebound}{B}

\newcommand{\psummarizedscatterfeat}{\tilde{\phi}}

\newcommand{\filterfun}{\mathbf{h}}
\newcommand{\scatternot}[1]{_{(#1)}}
\newcommand{\pathind}{p}

\newcommand{\scattertransform}{\bm {\phi}}
\newcommand{\transpose}{^\top}

\newcommand{\waveletnot}[1]{_{#1}}

\newcommand{\nodenodenot}[2]{_{#1#2}}

\newcommand{\layernot}[1]{^{(#1)}}

\newcommand{\nodeind}{n}

\newcommand{\scatteredfeatvec}{\mathbf{z}}

\newcommand{\gftscatteredfeatvec}{\check{\mathbf{z}}}

\newcommand{\identitymat}{\mathbf{I}}

\usepackage{amsthm}
\newtheorem{theorem}{Theorem}
\newtheorem{corollary}{Corollary}
\newtheorem{lemma}{Lemma}

\usepackage{xr}
\usepackage{wrapfig}

\usepackage{wrapfig}
\usepackage{tikz}
\usetikzlibrary{shapes,arrows}
\usetikzlibrary{fit,chains}

\usepackage{array}
\usepackage{cite}

\newcolumntype{H}{>{\setbox0=\hbox\bgroup}c<{\egroup}@{}}

\newfloatcommand{capbtabbox}{table}[][\FBwidth]

\title{Efficient and Stable Graph Scattering Transforms via Pruning
}
\author{%
Vassilis~N.~Ioannidis,~\IEEEmembership{Student~Member,~IEEE,}
Siheng~Chen,~\IEEEmembership{Member,~IEEE,}
and~Georgios~B.~Giannakis,~\IEEEmembership{Fellow,~IEEE,}
\IEEEcompsocitemizethanks{\IEEEcompsocthanksitem V.~N.~Ioannidis and G.~B.~Giannakis are with the Department
of Electrical and Computer Engineering, University of Minnesota, Minneapolis, MN. Emails: {$\{$ioann006,georgios$\}$@umn.edu}. S.~Chen is with Mitsubishi Electric Research Laboratories, Cambridge, MA. 
Email: schen@merl.com.  The work was carried out during V. N. Ioanndis' internship at Mitsubishi Electric Research Laboratories in the summer of 2019.}
}

\begin{document}
\maketitle
\begin{abstract}
Graph convolutional networks (GCNs) have well-documented performance in various graph learning tasks, but  their analysis is still at its infancy. Graph scattering transforms (GSTs) offer training-free deep GCN models that extract features from graph data, and are amenable to generalization and stability analyses. The price paid by GSTs is exponential complexity in space and time that increases with the number of layers. This discourages deployment of GSTs when a deep architecture is needed. The present work addresses the complexity limitation of GSTs by introducing an efficient so-termed pruned (p)GST approach. The resultant pruning algorithm is guided by a graph-spectrum-inspired criterion, and retains informative scattering features on-the-fly while bypassing the exponential complexity associated with GSTs. Stability of the novel pGSTs is also established when the input graph data or the network structure are perturbed. Furthermore, the sensitivity of pGST to random and localized signal perturbations is investigated analytically and experimentally. Numerical tests showcase that pGST performs comparably to the baseline GST at considerable computational savings. Furthermore, pGST achieves comparable performance to state-of-the-art GCNs in graph and 3D point cloud classification tasks. Upon analyzing the pGST pruning patterns, it is shown that graph data in different domains call for different network architectures, and that the pruning algorithm may be employed to guide the design choices for contemporary GCNs. 
\end{abstract}
\vspace{0mm}
\section{Introduction}
\vspace{0mm}
The abundance of graph-abiding data calls for advanced learning techniques that complement nicely standard machine learning tools when the latter cannot be directly employed, e.g. due to irregular data inter-dependencies. Permeating the benefits of deep learning to graph data, graph convolutional networks (GCNs) offer a versatile and powerful framework to learn from complex graph data~\cite{bronstein2017geometric}. GCNs and variants thereof have been remarkably successful in social network analysis, 3D point cloud processing, recommender systems and action recognition. However, researchers have recently reported less consistent perspectives on the desirable GCN designs. For example, experiments in social network analysis have argued that deeper GCNs marginally increase the learning performance~\cite{wu2019simplifying}, whereas a method for 3D point cloud segmentation achieves state-of-the-art performance with a 56-layer GCN network~\cite{li2019can}. These `controversial' empirical findings motivate theoretical analysis to understand the fundamental performance-defining factors, and the resultant design choices for high-performance GCNs.

Aiming to bestow GCNs with theoretical guarantees, one promising path is to study graph scattering transforms (GSTs) -- an analysis framework that has been advocated to assess stability and explain the success of deep neural networks (DNNs)~\cite{bruna2013invariant,mallat2012group}. GSTs are non-trainable GCNs comprising a cascade of graph filter banks followed by nonlinear activation functions. The graph filter banks are designed analytically to scatter an input graph signal into multiple channels. GSTs extract stable features of graph data that can be utilized for downstream graph learning tasks~\cite{gao2019geometric}, with competitive performance especially when the number of training examples is small. Under certain conditions on the graph filter banks, GSTs are endowed with energy conservation properties~\cite{zou2019graph}, as well as stability that amounts to robustness to graph topology deformations~\cite{gama2019stability}.  Inherited from scattering transforms, GSTs however are known to incur exponential complexity in space and time that increases with the number of layers 
(GST depth)~\cite{bruna2013invariant,mallat2012group}. Furthermore, stability should not come at odds with sensitivity. A filter’s output should be sensitive to and able to cope with perturbations of large magnitude. Current GST efforts have not addressed transform sensitivity to input noise. Lastly, graph data in different domains (e.g., social networks versus 3D point clouds) have distinct properties, which prompts domain-adaptive GST designs. 

\begin{figure*}
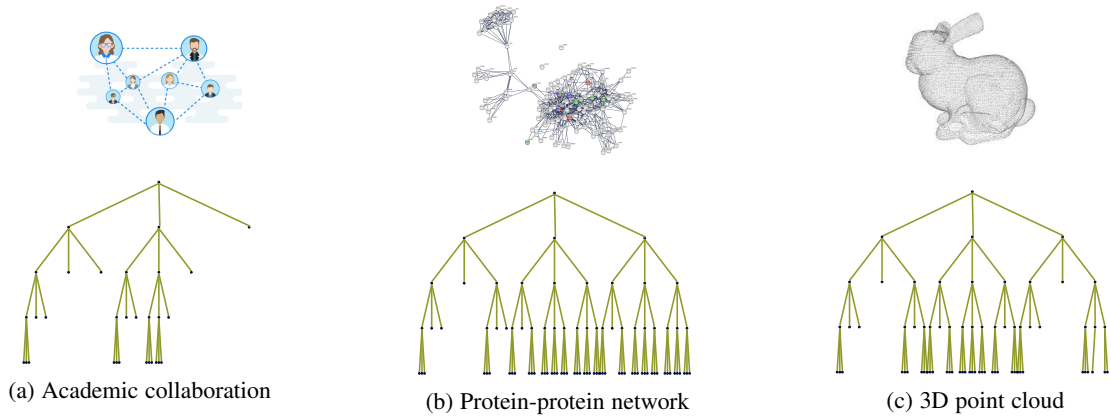

\begin{subfigure}{0.3\textwidth}
\centering\input{figs/collabpgstree.tex}
\caption{Academic collaboration}
\end{subfigure}
\begin{subfigure}{0.3\textwidth}
\centering\input{figs/protein.tex}
\caption{Protein-protein network}
\end{subfigure}
\begin{subfigure}{0.3\textwidth}
\centering\input{figs/pointcloud.tex}
\caption{3D point cloud}
\end{subfigure}
    \caption{Illustration of the same pGST applied to different graph datasets. For the social network (a) most GST branches are pruned, suggesting that most information is captured by local interactions.}
     \label{fig:pGSTex}
\vspace{0mm}
\end{figure*}
\vspace{0mm}

\vspace{0mm}
\subsection{Contributions}
The present paper develops a data-adaptive pruning approach to systematically retain informative GST  features, which justifies the term pruned graph scattering transform (pGST). The pruning decisions are guided by a criterion promoting alignment (matching) of the input graph spectrum with that of the graph filters. The optimal pruning decisions are provided on-the-fly, and alleviate the exponential complexity of GSTs. We prove that the pGST is stable to perturbations of the input graph data and those of the network structure. Under certain conditions on the perturbation energy, the resultant pruning patterns before and after the network and input perturbations are identical, and the overall pGST is stable. Further, the sensitivity of pGST to random and localized noise is theoretically and experimentally investigated. It turns out that pGST is more sensitive to noise that is localized in the graph spectrum relative to noise that is uniformly spread over the spectrum. This is appealing because pGST can detect transient changes in the graph spectral domain, while ignoring small random perturbations. With extensive experiments we showcase that the proposed pGSTs perform similar to, and in certain cases better than, the baseline GSTs that use all scattering features, while achieving significant computational savings. Furthermore, the \emph{efficient} and \emph{stable} features extracted can be utilized towards graph classification and 3D point cloud recognition. Even without any training on the feature extraction step, the performance is comparable to state-of-the-art deep supervised learning approaches, particularly when training data are scarce. By analyzing the pruning patterns of the pGST, we deduce that graph data in different domains call for distinct network architectures; see Fig.~\ref{fig:pGSTex}. Finally, we establish that the number of pGST layers after pruning can be utilized to guide the design parameters of contemporary GCNs.

A preliminary version of this work was presented in~\cite{ioannidis2020pgsticlr}. Relative to this, the novelty here is fourfold. First, novel theoretical results are derived that establish stability of the pGST to structural perturbations. An additional interpretation of the pruining criterion as a rate-distortion tradeoff is presented. Furthermore, sensitivity analysis of pGST reveals that the transform is more sensitive to noise that is localized in the graph spectrum relative to noise that is uniformly spread over the spectrum, which is appealing because pGST can detect transient changes in the graph spectral domain. Finally, new experimental results showcase: i) the effect of localized and random noise to the pruning algorithm; ii) the competitive performance of pGST for semi-supervised learning tasks; and, iii) the strong link between design choices in GCNs, and the pruning patterns of pGST.

\section{Related work}
\subsection{Graph convolutional networks}
GCNs rely on a layered processing architecture comprising trainable graph convolutional operations to linearly combine features per graph neighborhood, followed by pointwise nonlinear functions applied to the linearly transformed features~\cite{bronstein2017geometric}. Complex GCNs and their variants have shown remarkable success in graph semi-supervised learning~\cite{kipf2016semi,velivckovic2017graph} and graph classification tasks~\cite{ying2018hierarchical}. GCNs as simple as a single-layer linear module can offer high performance in certain social network learning applications~\cite{wu2019simplifying}. On the other hand, a 56-layer GCN has been employed to achieve state-of-the-art performance in 3D point cloud segmentation~\cite{li2019can}. Whether simple or complex, designing GCNs guided by properties of the graph data at hand is of paramount importance. 

\subsection{Graph scattering transforms}
To understand the success of GCNs analytically, recent works study the stability properties of GSTs with respect to metric deformations of the domain \cite{gama2018diffusion,gama2019stability,zou2019graph}. GSTs generalize scattering transforms~\cite{bruna2013invariant,mallat2012group} to non-Euclidean domains. 
GSTs are a cascade of graph filter banks and nonlinear operations that is organized in a tree-structured architecture. The number of extracted GST features grows exponentially with the number of layers. Theoretical guarantees for GSTs are obtained after fixing the graph filter banks to implement a set of graph wavelets. The work in \cite{zou2019graph} establishes energy conservation properties for GSTs given that certain energy-preserving graph wavelets are employed, and also prove that GSTs are stable to graph structure perturbations; see also \cite{gama2018diffusion} that focuses on diffusion wavelets. On the other hand, \cite{gama2019stability} proves stability to relative metric deformations for a wide class of graph wavelet families. These contemporary works shed light into the stability and generalization capabilities of GCNs. However, stable transforms are not necessarily informative, and albeit highly desirable, a principled approach to selecting informative GST features remains still an uncharted venue.

\subsection{Neural network based compression}
This work advocates pruning of GSTs to avoid the exponential growth of features with the network depth, which is naturally related to deep neural network (DNN) based compression~\cite{reed1993pruning,han2015deep}. State-of-the-art DNNs typically entail a large number of network layers and corresponding trainable parameters, which introduce excessive computational and memory costs~\cite{krizhevsky2012imagenet,he2016deep}. Confronting these challenges, neural network compression aims at reducing the network parameters and pruning the architecture to facilitate practical deployment of DNN-based solutions. Typical compression techniques prune redundant parameters during training, while minimizing the effect of the pruned parameters to the learned features~\cite{reed1993pruning,han2015deep}. Recent approaches perform network pruning at initialization before training the DNN~\cite{lee2018snip}. 

Although DNN-based compression is a fruitful direction, pruning GCNs has not received commensurate attention. Typical GCNs are applied to relatively small data sets and require only a small number of layers and training parameters to attain state-of-the-art learning performance~\cite{kipf2016semi, velivckovic2017graph,wu2019simplifying}, which explains why GCN-based compression is yet to be explored. Deep GCNs are emerging however, to deal with web-scale graphs~\cite{li2019can,ying2018graph}. Such overparameterized GCNs motivate the development of novel compression techniques. Our work can be seen as a stepping stone towards pruning GCNs; meanwhile, pruning GSTs versus GCNs/DNNs has differences in two aspects. First, GSTs are nontrainable feature extractors, whereas GCNs/DNNs introduce trainable parameters that should be considered by the compression algorithm. Second, GST pruning is inherently an one-shot process and pruning is performed in an online and even adaptive fashion, whereas pruning GCNs or DNNs should be performed offline.

\vspace{0mm}
\section{Background}
\vspace{0mm}
Consider an undirected graph $\mathcal{G}\define\{\mathcal{V},\mathcal{E}\}$ with node set $\mathcal{V}\define\{v_i\}_{i=1}^N$, and edge set $\mathcal{E}\define\{e_i\}_{i=1}^E$. Its connectivity is described by the~\emph{graph shift matrix} $\mathbf{S} \in \mathbb{R}^{N\times N}$, whose $(n,n')$th entry $S\nodenodenot{n}{n'}$ is nonzero if $(n,n') \in \mathcal{E}$ or if  $n=n'$. A typical choice for $\mathbf{S}$ is the symmetric adjacency or the Laplacian matrix. Further, each node can be also associated with a few attributes. Collect attributes across all nodes in the matrix $\mathbf{X}\define[\mathbf{x}_{1},\ldots,\mathbf{x}_{F}] \in  \mathbb{R}^{N \times F}$, where each column $\mathbf{x}_{f} \in \mathbb{R}^N$ can be regarded as a~\emph{graph signal}. 

\subsection{Graph Fourier transform}
A Fourier transform corresponds to the expansion of a signal over bases that are invariant to filtering; here, this graph frequency basis is the eigenbasis of the shift matrix $\mathbf{S}$. Henceforth,   $\mathbf{S}$ is assumed normal with $\mathbf{S}= \mathbf{V} \bm{\Lambda} \mathbf{V}\transpose,$ where  $\mathbf{V} \in\mathbb{R}^{N\times N}$ forms the graph Fourier basis,  and
$\bm{\Lambda} \in \mathbb{R}^{N\times N}$ is the diagonal matrix of corresponding eigenvalues $\lambda_0, \, \ldots, \, \lambda_{N-1}$, that can be thought of as graph frequencies.  The graph Fourier transform (GFT) of $\mathbf{x} \in \mathbb{R}^N$ is $\widehat{\mathbf{x}} = \mathbf{V}\transpose \mathbf{x} \in \mathbb{R}^{N}$, while the inverse transform is $\mathbf{x}  =  \mathbf{V}  \widehat{\mathbf{x}}$. The vector $\widehat{\mathbf{x}}$ represents the signal's expansion in the eigenvector basis and describes the graph spectrum of  $\mathbf{x}$. The inverse GFT reconstructs the graph signal from its graph spectrum by combining graph frequency components weighted by the coefficients of the signal's graph Fourier transform. GFT is a theoretical tool that has been popular for analyzing graph data in the graph spectral domain.

\subsection{Graph convolutional networks}
GCNs permeate the benefits of CNNs from processing Euclidean data to modeling graph structured data.  GCNs model graph data through a succession of layers, each consisting of a graph convolutional operation (a.k.a. graph filter), a pointwise nonlinear function $\sigma(\cdot)$, and oftentimes also a pooling operation. Given $\mathbf{x} \in  \mathbb{R}^{N}$, the graph convolution operation diffuses each node's information to its neighbors to obtain $\mathbf{S}\mathbf{x}$ with $n$th entry  $[\mathbf{S}\mathbf{x}]_n=\sum_{n'\in \mathcal{N}_n }S_{nn'}x_{n'}$ being a weighted average of the one-hop neighboring features.  Successive application of $\mathbf S$ will reach multi-hop neighbors, spreading the information across the network. Summing up, a $K$th-order graph convolutional operation (graph filtering) is
\begin{align}
\label{eq:graph_conv}
\vspace{1mm}
{h}(\mathbf{S}) \mathbf{x}:= \sum_{k=0}^{K}w_k\mathbf{S}^k \mathbf{x} = \mathbf{V}  \widehat{h}(\bm{\Lambda}) \widehat{\mathbf{x}}
\vspace{1mm}
\end{align}
where the graph filter ${h}(\cdot)$ is parameterized by the learnable weights $\{w_k\}_{k=0}^{K}$, and in the graph spectral domain $\widehat{h}(\bm{\Lambda}) = \sum_{k=0}^{K} w_k \bm{\Lambda}^k$. 
In the graph vertex domain, the learnable weights capture the influences from various orders of neighbors; and in the graph spectral domain, those weights adaptively adjust the focus and emphasize certain graph frequency bands. GCNs employ various graph filter banks per layer, and learn the parameters that minimize a predefined learning objective, such as classification, or regression. 

\subsection{Graph  scattering transforms}
GSTs are the training-free counterparts of GCNs, where the parameters of graph convolutions are fixed based on a design criterion. Per GST layer, the input is graph filtered using a filter bank $\{{h}_j(\mathbf{S})\}_{j=1}^{J}$, an elementwise nonlinear function $\nonlinearity{\cdot}$, and a pooling operator $ U$. At the first layer, the input $ \mathbf{x} \in \mathbb{R}^{N}$ constitutes the first scattering feature vector $\mathbf{z}_{(0)}: = \mathbf{x}$. Next, $\mathbf{z}_{(0)}$ is processed through 
$\{{h}_j(\cdot)\}_{j=1}^{J}$ and $\sigma(\cdot)$ to generate $\{\mathbf{z}_{(j)}\}_{j=1}^{J}$ with $\mathbf{z}_{(j)}:= \sigma({h}_j(\mathbf{S})\mathbf{z}_{(0)})$. At the second layer, the same operation is repeated per $j$. This yields a tree structure with ${J}$ branches stemming out from non-leaf node; see also Fig.~\ref{fig:gst}. The $\ell$th layer  of the tree includes $J^{\ell}$ nodes. Each node of layer $\ell$ is indexed by the path $p^{(\ell)}$ of the sequence of $\ell$ graph convolutions applied to the input $\mathbf{x}$, i.e. $p^{(\ell)}\define(j^{(1)},j^{(2)},\ldots,j^{(\ell)})$.\footnote{A tree node is fully specified by its corresponding path.}
\begin{figure*}
    \centering
    \input{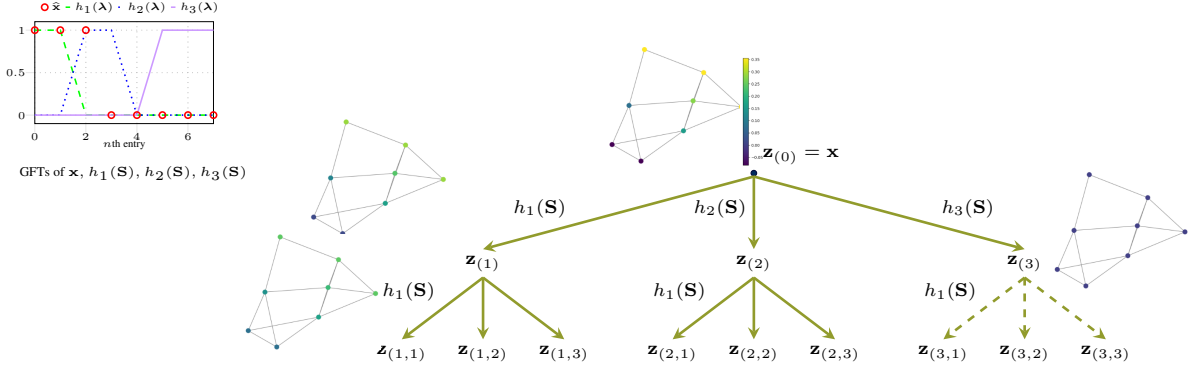}
    \caption{Scattering pattern of a pGST with $J=3$ and $L=3$. Dashed lines represent pruned branches.
    An example $\mathbf{x}$ and GFTs of filter banks are depicted too.  The third filter $j=3$ at $\ell=1$ is pruned because it generates no output ($\mathbf{z_{(3)}}=\mathbf{0}$). 
    }
    \label{fig:gst}
    \vspace{1mm}
\end{figure*}
The scattering feature vector at the tree node indexed by ${(p^{(\ell)},{j})}$ at layer $\ell+1$, is
\begin{align}
\label{eq:convoper}
    \mathbf{z}\scatternot{p^{(\ell)},{j}}=\nonlinearity{{h}_j(\mathbf{S})\mathbf{z}_{(p^{(\ell)})}}
\end{align}
where $p^{(\ell)}$ holds the list of indices of the parent nodes ordered by ancestry, and all path $p^{(\ell)}$ in the tree with length $\ell$ are included in the path set $\mathcal{P}^{(\ell)}$ with $|\mathcal{P}^{(\ell)}|=2^{\ell}$. The memoryless nonlinearity $\sigma(\cdot)$ disperses the graph frequency representation through the spectrum, and endows the GST with increased discriminating power~\cite{gama2019stability}.  By exploiting graph sparsity, the  computational complexity of \eqref{eq:convoper} is $\mathcal{O}(KE)$, where $E = |\mathcal{E}|$ is the number of edges in\footnote{Any analytical function $h(\mathbf{S})$ can be written as a polynomial of $\mathbf{S}$ with maximum degree $N-1$ \cite{horn2012matrix}.} $\mathcal{G}$. Each feature vector $\mathbf{z}_{(p^{(\ell)})}$ is summarized by an aggregation operator $U(\cdot)$ defining a scalar scattering coefficient $\phi_{(p^{(\ell)})}\define U(\mathbf{z}_{(p^{(\ell)})})$, where $U(\cdot)$ is typically an averaging operator that reduces dimensionality of the extracted features. The scattering coefficient per tree node reflects the activation level at a certain graph frequency band.

These scattering coefficients are collected across all tree nodes to form a scattering feature map
\begin{align}
\label{eq:defscatvect}
    \bm{\Phi}(\mathbf{x}):=\left\{\{ \phi_{(p^{(\ell)})}\}_{p^{(\ell)}\in\mathcal{P}^{(\ell)}}\right\}_{\ell=0}^{L}
\end{align}
where  $|\bm{\Phi}(\mathbf{x})| = \sum_{\ell=0}^L{J}^\ell$.  
The GST operation resembles  a forward pass of a trained GCN. This is why several works study GST stability under perturbations of $\mathbf{S}$ in order to understand the working mechanism of GCNs~\cite{zou2019graph,gama2019stability,gama2018diffusion}.
\vspace{0mm}
\section{Pruned Graph Scattering Transforms}
\vspace{0mm}
While the representation power of GST increases with the number of layers, the computational and storage complexity of the transform also increase exponentially with the number of layers due to its scattering nature. Hence, even if informative features become available at deeper layers, the associated exponential complexity of extracting such features is prohibitive. On the other hand, various input  data (e.g., social networks and 3D point clouds) may have  distinct properties, leading to different GST feature maps. In some cases, just a few nodes of deep layers are informative; and in other cases, nodes of shallow layers convey most of the information; see Fig.~\ref{fig:pGSTex}. This motivates the pursuit of a tuned GST to adaptively capture informative nodes.

Aspiring to improve GST models, we introduce a~\emph{pruned} (p)GST to systematically retain informative nodes without extra complexity. Our novel pGST reduces the exponential complexity and adapts GST to different graph data. Furthermore, pGST offers a practical mechanism to understand the architecture of GCNs. Based on the pruning patterns, the proposed pGST suggests when a deeper GCN is desirable, or, when a shallow one will suffice. Pruning the wavelet packets has been traditionally employed for compression in image processing applications~\cite{xiong2002space}, where the pruning is guided by a rate-distortion optimality criterion. 

In this work, we consider a graph spectrum inspired criterion. Intuitively, each GST node will be associated with a unique subband of the graph spectrum. When the subband of a node does not have sufficient overlap with the graph signal spectrum, this node cannot capture the underlying graph signal, and should be pruned. Consider for specificity a smooth $\mathbf{x}$, where connected nodes have similar signal values. This implies sparse (low-rank) representation in the graph spectral domain; that is, $\widehat{\mathbf{x}}\define\mathbf{V}^\top\mathbf{x} \in \mathbb{R}^N$ and $[\widehat{\mathbf{x}}]_{n}=0$ for $n \geq b$. The graph spectrum of the $j$th output is then
\begin{align}
\nonumber
    \mathbf{V}\transpose{h}_j(\mathbf{S})\mathbf{x}=&\diag{(\widehat{{h}}_j(\bm{\lambda}))}\widehat{\mathbf{x}}\\
    =&[\widehat{{h}}_j(\lambda_1)
\widehat{x}_1,
\widehat{{h}}_j(\lambda)\widehat{x},\ldots,\widehat{{h}}_j(\lambda_N)\widehat{x}_N]\transpose
\end{align}
where $\lambda_n$ is the $n$th eigenvalue of $\mathbf{S}$ and each frequency $\widehat{x}_n$ is weighted by the corresponding transformed eigenvalue $\widehat{{h}}_j(\lambda_n)$.
Hence, if the support of the spectrum  $\{\widehat{{h}}_j(\lambda_n)\}_n$ is not included in the support of $[\widehat{\mathbf{x}}]_{n}$, the $j$th graph filter output will not capture any information; that is, ${h}_j(\mathbf{S})\mathbf{x}=\mathbf{0}_N$; see Fig. \ref{fig:gst}.  Thus, identifying such graph filters and pruning the corresponding tree nodes will result in a parsimonius and thus computationally efficient GST.

\subsection{Pruning criterion}
Motivated by this last observation, we introduce a pruning criterion to select the scattering branches per node by maximizing the alignment between the graph spectrum of the filters and the scattering features. Per node $p$, the optimization problem is
\begin{align}
\label{eq:criterion}
    \max_{\{ f_{(p, j)} \}_{j=1}^J}&~~~~~
    \sum^{J}_{j=1}
\left(\sum_{n=1}^N\left(\widehat{{h}}_j(\lambda_n)^2-\tau\right)[\widehat{\mathbf{z}}_{(p)}]^2_n\right) {f_{(p, j)}}
\\ \nonumber
    \text{s. t.}&~~~~~f_{(p, j)} \in \{0,1\},~~j=1,\ldots,{J}
\end{align}
where $\widehat{\mathbf{z}}_{(p)} \define \mathbf{V} \mathbf{z}_{(p)}$ is the graph spectrum of the scattering feature vector $\mathbf{z}_{(p)}$; $\tau$ is a user-specific threshold; and, $f_{(p, j)}$ stands for the pruning assignment variable indicating whether node $(p, j)$ is active  ($f_{(p, j)}=1$), or, it should be pruned ($f_{(p, j)}=0$). The objective in~\eqref{eq:criterion} promotes retaining tree nodes that maximize the alignment of the graph spectrum of $\widehat{\mathbf{z}}_{(p)}$ with that of $\widehat{h}_j{(\bm{\lambda})}$. The  threshold $\tau$ introduces a minimum spectral value to locate those nodes whose corresponding graph spectral response is small, meaning $\widehat{{h}}_j(\lambda_n)^2 \ll \tau$. Note that criterion \eqref{eq:criterion} is evaluated per tree node $p$, thus allowing for a flexible and scalable design.

The optimization problem in \eqref{eq:criterion} is nonconvex since $f_{(p, j)}$ is a discrete variable. Furthermore, recovering $\widehat{\mathbf{z}}_{(p)}$ requires an eigendecomposition of the Laplacian matrix that incurs complexity $\mathcal{O}(N^3)$. Nevertheless, by exploiting the structure in \eqref{eq:criterion}, we will develop an efficient pruning algorithm attaining the maximum of $\eqref{eq:criterion}$, as asserted by the following theorem.
\begin{theorem}\label{th:pruncond}
The optimal pruning assignment variables $\left\{f_{(p, j)}^*\right\}_j$  of \eqref{eq:criterion} are given by 
\begin{align}
\label{eq:pwcriterion}
        f_{(p, j)}^* =\begin{dcases*} 1
   & if  $\frac{\|\mathbf{z}_{(p, j)}\|^2}
{\|\mathbf{z}_{(p)}\|^2}>\tau,$ \\[0.5ex]
0
   & otherwise.
\end{dcases*},~~~j=1,\ldots,{J}
\vspace{1mm}
\end{align}
\end{theorem}
The optimal variables $f_{(p, j)}^*$ are given by comparing the energy of the input $\mathbf{z}_{(p)}$ to that of the output $\mathbf{z}_{(p, j)}$ per graph filter $j$ that can be evaluated with linear complexity of order $\mathcal{O}(N)$. Our pruning criterion leads to a principled and scalable means of selecting the GST nodes to be pruned. The pruning objective is evaluated per node $p$, and pruning decisions are made on-the-fly. Hence, when $f_{(p)}^*=1$, node $p$ is active and the graph filter bank will be applied to $\mathbf{z}_{(p)}$, expanding the tree to the next layer; otherwise, the GST will not be expanded further at node $p$, which can effect exponential savings in computations.  An example of such a pruned tree is depicted in Fig.~\ref{fig:gst}.  Evidently, the hyperparameter $\tau$ controls the input-to-output energy ratio.  A large $\tau$ corresponds to an aggressively pruned scattering tree, while a small $\tau$ amounts to a minimally pruned scattering tree. The stable and efficient representation extracted by the pGST is then defined as 
\begin{align}\nonumber
     \bm{\Psi}(\mathbf{x}): =\left\{\phi_{(p)}\right\}_{p\in\mathcal{T}}
\vspace{1mm}
\end{align}
where $\mathcal{T}$ is the set of active tree nodes $\mathcal{T}:=\{p \in\mathcal{P}|~f_{(p)}^*=1 \}$.

Our pruning approach provides a concise version of GSTs and effects savings in computations as well as memory. Although the worst-case complexity of pGST is still exponential, a desirable complexity can be realized by properly selecting $\tau$. As a byproduct, the scattering patterns of pGSTs reveal the appropriate depths and widths of the GSTs for different graph data; see also Fig.~\ref{fig:pGSTex}. The pruning approach so far is an unsupervised one, since the labels are not assumed available. Note that the GFT is used throughout for analytical purposes, and neither the pruning algorithm nor the pGST requires explicit calculation of the GFT.

\subsection{Rate-distortion tradeoff}
Our pruning objective is closely related to the optimal rate-distrortion objective in signal processing~\cite{tsatsanis1995principal,ramchandran1993best}.
By judiciously selecting $\tau$, the proposed pruning criterion in \eqref{eq:pwcriterion} finds also the optimal variables to the following problem
\begin{align}
\label{eq:criterionalt}
    \max_{\{ f_{(p, j)} \}_{j=1}^J}&~~~~~
    \sum^{J}_{j=1}\frac{
\|{\mathbf{z}}_{(p,j)}\|^2} 
{\|{\mathbf{z}}_{(p)}\|^2}f_{(p, j)}
\\ \nonumber
    \text{s. t.}&~~~~~f_{(p, j)} \in \{0,1\},~~j=1,\ldots,{J}
    \\ &~~~~~\sum_{j=1}^J f_{(p, j)} \le K\nonumber
\end{align}
where $K$ denotes the maximum number of retained scattering features, that can be dictated by computational complexity constraints.
The optimal solution to~\eqref{eq:criterionalt} is to set $f_{(p, j)}=1$ for the $K$ scattering features with larger energy ratio ${
\|{\mathbf{z}}_{(p,j)}\|^2}/
{\|{\mathbf{z}}_{(p)}\|^2}$ and the rest to zero $f_{(p, j)}=0$.
Indeed, by properly choosing $\tau$, the criterion in \eqref{eq:pwcriterion} can be employed to find the $K$ channels with maximum energy, where $\tau$ would be the $K$th larger energy ratio. 
Such energy preservation objective has been employed in signal processing in the context of PCA wavelets~\cite{tsatsanis1995principal}. Indeed, the hyperparameter $\tau$ reflects the sweet spot in the rate-distortion tradeoff~\cite{ramchandran1993best}. 

\section{Stability and sensitivity of  pGST}
\vspace{0mm}
In this section, we study the stability and sensitivity of pGST when the input graph data and the network topology are perturbed. We will establish that pGST is stable in the presence of feature or network perturbations with bounded power. We will further analyze pGST sensitivity to input perturbations. To establish these results, we consider graph wavelets that form a frame with frame bounds $A$ and $B$~\cite{hammond2011wavelets}, meaning for $\mathbf{x} \in \mathbb{R}^N$, it holds that, $A^2\|\mathbf{x}\|^2\le\sum_{j=1}^J\|{h}_j(\mathbf{S})\mathbf{x}\|^2\le B^2\|\mathbf{x}\|^2.$ In the graph vertex domain, $A$ and $B$ characterize the numerical stability of recovering $\mathbf{x}$ from $\{{h}_j(\mathbf{S})\mathbf{x}\}_{j}$. In the graph spectral domain, they reflect the ability of a graph filter bank to amplify $\mathbf{x}$ along each graph frequency. Tight frame bounds, satisfying $A^2=B^2$, are of particular interest because such wavelets lead to enhanced numerical stability and faster computations~\cite{shuman2015spectrum}. The frame property of a graph wavelet plays an instrumental role in proving GST stability to perturbations of the underlying graph structure~\cite{gama2019stability,gama2018diffusion,zou2019graph}.

\subsection{Stability to graph data perturbations} 
Perturbations present in $\mathbf{x}$ may be attributed to modeling errors, or adversarial intervention aiming to poison learning. Consider the perturbed graph signal
\begin{align}
\label{eq:simplvecpert}
    \tilde{\mathbf{x}}\define\mathbf{x}+{\bm{\delta}} \in \mathbb{R}^N
\end{align}
where $\bm{\delta} \in \mathbb{R}^N$ is the perturbation vector. 
We wish to study how and under what conditions our pGST is affected by such perturbations. A stable transformation should have a similar output under small input perturbations. 

Before establishing that our pGST is stable, we first show that GST is stable to small perturbations in
$\mathbf{x}$. Prior art deals mainly with GST stability to structure perturbations~\cite{gama2019stability,gama2018diffusion,zou2019graph}.
\begin{lemma}\label{th:stab1}
Consider the GST $\bm{\Phi}(\cdot)$ with $L$ layers and $J$ graph filters; and suppose that the graph filter bank forms a frame with bound $B$, while $\mathbf{x}$ and $\tilde{\mathbf{x}}$ are related via \eqref{eq:simplvecpert}. It then holds that
\begin{align}
    \frac{\|\bm{\Phi}(\mathbf{x})-\bm\Phi (\tilde{\mathbf{x}})\|}{\sqrt{|\bm{\Phi}(\mathbf{x})|}}  \le 
    \sqrt{\frac{\sum_{\ell=0}^{L}(\framebound^{2}{J})^\ell}{\sum_{\ell=0}^L{J}^\ell}}
    \|{\bm{\delta}}\|~. \label{eq:gstbound}
\vspace{-4mm}
\end{align}
\end{lemma}
The squared difference of the GSTs is normalized by the number of scattering features in $\bm{\Phi}(\cdot)$, that is $|\bm{\Phi}(\mathbf{x})| = \sum_{\ell=1}^L{J}^\ell$. The bound in~\eqref{eq:gstbound} relates to the frame bound of the wavelet filter bank.  Clearly, for tight frames with $B=1$, the normalized stability bound~\eqref{eq:gstbound} is tight. Let $\tilde{\mathcal{T}}$ be the structure of the pruned tree for $\bm{\Psi}({\tilde{\mathbf{x}}})$. The following lemma asserts that the pGST offers the same pruned tree for the original and the perturbed inputs. 
\begin{lemma}\label{th:stab2}
Let $\bm{\Psi}(\cdot)$ denote the pGST with $L$ layers and $J$ graph filters; $\tilde{\mathbf{z}}_{p}$ the perturbed scattering feature at node $p$; and, ${\bm{\delta}}_{p}:=\mathbf{z}_{p}-\tilde{\mathbf{z}}_{p}$. If for all $p\in\mathcal{P}$ and $j=1,\ldots,{J}$, we have  
\begin{align}
\big|\|{h}_j(\mathbf{S})\mathbf{z}_{p}\|^2
-\tau\|\mathbf{z}_{p}\|^2\big|> \|{h}_j(\mathbf{S}){\bm{\delta}}_{p}\|^2+
    \tau\big|
    \|\mathbf{z}_{p}\|^2-\|\tilde{\mathbf{z}}_{p}\|^2\big|.\label{eq:thassump}
\end{align} 
it then follows that
\vspace{-0.2cm}\begin{myitemize}
\addtolength{\itemindent}{-0.5cm}
    \item[i)] The pruned scattering transform will output the same tree for $\bm{\Psi}({\mathbf{x}})$ and $\bm{\Psi}({\tilde{\mathbf{x}}})$; that is, $\mathcal{T}=\tilde{\mathcal{T}}$; and, 
    \item[ii)] With $g(\mathbf{x})\define\|{h}_j(\mathbf{S})\mathbf{x}\|^2
    -\tau\|\mathbf{x}\|^2$, a necessary condition for \eqref{eq:thassump} is 
    \begin{align}
    \label{eq:fincond}
|g(\mathbf{z}_{p})|>g({\bm{\delta}}_{p})\;.
\vspace{1mm}
\end{align} 
\end{myitemize}
\end{lemma}
According to~\eqref{eq:fincond}, Lemma~\ref{th:stab2} can be interpreted as a signal-to-noise-ratio (SNR) condition because under  $g({\bm{\delta}}_{p})>0$, it is possible to write~\eqref{eq:fincond} as $|g(\mathbf{z}_{p})|/g({\bm{\delta}}_{p})>1$.  Lemma~\ref{th:stab2} provides a per-layer and branch condition for pGST to output the same pruned scattering tree for the original or the perturbed signal. This condition is also experimentally validated in Sec.~\ref{sec:exp}, where the structure in $\mathcal{T}$ remains the same for Fig.~\ref{fig:ran} and Fig.~\ref{fig:orig}.

By combining Lemmas \ref{th:stab1} and \ref{th:stab2}, we arrive at the following stability result for the pGST network. 
\begin{theorem}\label{th:stab}
Consider the pGST transform $\bm{\Psi}(\cdot)$ with $L$ layers and $J$ graph filters; and suppose that the graph filter bank forms a frame with bound $B$, while $\mathbf{x}$ and $\tilde{\mathbf{x}}$ are related via \eqref{eq:simplvecpert}. The pGST is stable to bounded perturbations ${\bm{\delta}}$, in the sense that
\begin{align}
\vspace{1mm}
\label{eq:stabfeat}
   \frac{\|\bm{\Psi}(\mathbf{x})-\bm{\Psi} (\tilde{\mathbf{x}})\|}{\sqrt{|\bm{\Psi}(\mathbf{x})|}} 
    \le
   \sqrt{ \frac{
    \sum_{\ell=0}^{L}F_{\ell}\framebound^{2\ell}
    }{\sum_{\ell=0}^{L}F_{\ell}}}\|{\bm{\delta}}\|
\vspace{1mm}
\end{align}
where $F_\ell:=|\mathcal{P}^{(\ell)}\cup \mathcal{T}|$ is the number of active scattering features at layer $\ell$, and $|\bm{\Psi}(\mathbf{x})|=\sum_{\ell=0}^{L}F_{\ell}$ the number of retained scattering features.
\end{theorem}
The bound in~\eqref{eq:boundstab} is linear in the perturbation power $\|{\bm{\delta}}\|$, and hence pGST is stable to perturbations in the input.
\subsection{Stability to structural perturbations}
We next study the effect of network perturbations to the proposed pGST. The goal here is to establish that pGST is also stable under structural perturbations that can be due to noise or adversarial attacks
\cite{zugner18adv,goodfellow2014explaining}. Vanilla scattering transforms are invariant to translations and stable to perturbations that resemble translations~\cite{bruna2013invariant}.
Likewise, GSTs are invariant to permutations and stable to pertubations that are close to permutations~\cite{gama2019stability,zou2019graph}. In the graph context, permutations are regarded as rearrangements of the vertex indices. 

Consider a perturbed topology given by the $N \times N$ perturbed (and permuted) shift matrix $\tilde{\mathbf{S}}$. The set of permutations that make $\tilde{\mathbf{S}}$ close to $\mathbf{S}$ are given by 
\begin{align}
    \mathcal{P}_0:=\arg\min_{\mathbf{P}\in\mathcal{P}}\|\mathbf{P}\transpose\tilde{\mathbf{S}}\mathbf{P}-\mathbf{S}\|
\end{align}
where $\mathbf{P}$ is an $N \times  N$ permutation matrix having entries $0-1$. Further, consider the following set of perturbation matrices~\cite{gama2019stability}
\begin{align}
\label{eq:errorset}
    \mathcal{D}:=\bigg\{\bm{\Delta}:\mathbf{P}\transpose\tilde{\mathbf{S}}\mathbf{P}  = \mathbf{S}+\bm{\Delta}\transpose\mathbf{S} +
    \mathbf{S}\bm{\Delta}, \mathbf{P}\in\mathcal{P}_0\bigg\} 
\end{align}
where  $\bm{\Delta}$ denotes the $N\times N$ perturbation matrix. The error is captured by the term $\bm{\Delta}\transpose\mathbf{S} + \mathbf{S}\bm{\Delta}$. Given the set of perturbation matrices in~\eqref{eq:errorset}, we consider the distance between $\tilde{\mathbf{S}}$ and $\mathbf{S}$ as 
\begin{align}
    d(\mathbf{S},\tilde{\mathbf{S}})=\min_{\bm{\Delta}\in\mathcal{D}}\|\bm{\Delta}\|.
\end{align}
Without loss of generality, let $\mathbf{P}=\mathbf{I}$, otherwise fix a $\mathbf{P}_0 \in\mathcal{P}$ and define $\tilde{\mathbf{S}}$ to equal $\mathbf{P}_0\transpose\tilde{\mathbf{S}} \mathbf{P}_0$. Hence, the perturbed topology can be written as\footnote{Unless it is stated otherwise the $\ell_2$ norms are employed.}
\begin{align}
\tilde{\mathbf{S}}=\mathbf{S}+\bm{\Delta}\transpose\mathbf{S} +\mathbf{S}\bm{\Delta}.
\end{align}
Let $\tilde{\bm{\Psi}}(\cdot)$ denote the pruned scattering transform that is based on  the perturbed topology $\tilde{\mathbf{S}}$, and $\tilde{\mathcal{T}}$ be the structure of the pruned tree for $\tilde{\bm{\Psi}}({\mathbf{x}})$. The following lemma asserts that pGST outputs the same pruned tree for both the original and the perturbed graph. 

\begin{lemma}\label{th:stab3}Suppose $\mathbf{S}$ and $\tilde{\mathbf{S}}$ satisfy  $d(\mathbf{S},\tilde{\mathbf{S}})\le\mathscr{E}/2$; and that for  $\bm{\Delta}\in \mathcal{D}$ with eigendecomposition $\bm{\Delta}=\mathbf{U}\diag{(\mathbf{d})}\mathbf{U}\transpose$ it holds that  $\|\bm{\Delta}/d_{\textrm{max}}-\mathbf{I}\|\le\mathscr{E}$, where $d_{\textrm{max}}$ is the eigenvalue of $\bm{\Delta}$ with maximum absolute value. Suppose the graph filter bank forms a frame with bound $B$, and $h(\lambda)$ satisfies the integral Lipschitz constraint $|\lambda h'(\lambda)|\le C_0$. Let $\tilde{\mathbf{z}}_{p}$ denote the perturbed scattering feature at the tree node $p$ and  ${\bm{\delta}}_{p}:=\mathbf{z}_{p}-\tilde{\mathbf{z}}_{p}$. If for all nodes $p\in\mathcal{P}$ of layer $\ell$ and $j=1,\ldots,{J}$, it holds that 
\begin{align}
\big|\|{h}_j(\mathbf{S})\mathbf{z}_{p}\|^2
-\tau\|\mathbf{z}_{p}\|^2\big|>(\ell\mathscr{E} C_0 B^{\ell-1}\|\mathbf{x}\|)^2+
    \tau
    \|\bm{\delta}_{p}\|^2 \label{eq:thassump1}
\end{align} 
the pruned scattering transform will then output the same tree for $\bm{\Psi}({\mathbf{x}})$ and $\tilde{\bm{\Psi}}({\mathbf{x}})$; that is, $\mathcal{T}=\tilde{\mathcal{T}}$.
\end{lemma}

Lemma~\ref{th:stab3} establishes conditions under which the pGST based on the original and perturbed graphs will output the same scattering tree. The assumption on the eigenvalues of the pertubation, that is $\|\bm{\Delta}/d_{\textrm{max}}-\mathbf{I}\|\le\mathscr{E}$, limits the structural changes in the graph such as edge insertions or deletions; see also~\cite{gama2019stability}. The integral Lipschitz  constraint $|\lambda h'(\lambda)|\le C_0$ is an additional stability requirement on the wavelet $h(\cdot)$. Note that these conditions are not required to establish stability to input noise in Theorem~\ref{th:stab}.
By combining Lemma~\ref{th:stab3} with Proposition 3 in~\cite{gama2019stability}, we can show the stability of pGST to structural perturbations. 
\begin{theorem}
\label{th:stabst0}
Under the conditions of Lemma~\ref{th:stab3} it holds that
\begin{align}
\vspace{1mm}
\label{eq:boundstab}
   \frac{\|\bm{\Psi}(\mathbf{x})-\tilde{\bm{\Psi}} (\mathbf{x})\|}{\sqrt{|\bm{\Psi}(\mathbf{x})|}} 
    \le
   \mathscr{E} C_0\sqrt{ \frac{
    \sum_{\ell=0}^{L}F_{\ell}\ell^2\framebound^{2\ell}
    }{\sum_{\ell=0}^{L}F_{\ell}}}\|{\mathbf{x}}\|
\vspace{1mm}
\end{align}
where $F_\ell:=|\mathcal{P}^{(\ell)}\cup \mathcal{T}|$ is the number of active scattering features at layer $\ell$, and $|\bm{\Psi}(\mathbf{x})|=\sum_{\ell=0}^{L}F_{\ell}$ the number of retained scattering features.
\end{theorem}
The bound in~\eqref{eq:boundstab} is linear in the perturbation $\mathscr{E}$, which establishes the stability of pGST. As the bound depends also linearly in the Lipstitz constant $C_0$ and the number of layers $L$, these terms also appear in the stability results for traditional GST~\cite{gama2019stability}.
For tight frames with $B=1$, the bound in~\eqref{eq:boundstab} can be further simplified.
By contrasting the bounds in~\eqref{eq:stabfeat} and~\eqref{eq:boundstab}, we deduce that~\eqref{eq:stabfeat} is tighter than~\eqref{eq:boundstab} by a factor $\ell$. This follows because the topology perturbation in~\eqref{eq:boundstab} implies a perturbed wavelet that is repeatedly applied to the scattering features. Hence, the error in~\eqref{eq:boundstab} is accumulated layer by layer, giving rise to the $\ell$ factor in~\eqref{eq:boundstab}. On the other hand, for~\eqref{eq:stabfeat} the perturbation is introduced only at the first layer $\ell=0$. Theorems \ref{th:stab} and \ref{th:stabst0} can be combined to guarantee stability under \emph{joint} perturbations of both the input and the topology.

\subsection{Sensitivity to graph data perturbations} 
Next, the sensitivity of pGST is analyzed for different types of noises in input signals\footnote{The sensitivity analysis for structural perturbations does not provide interesting results; see the experimental validation in Section~\ref{sec:structural_sensitivity}. Highly localized signal perturbation would affect only a few scattering filters; however, arbitrary structural perturbation would affect all scattering filters $\{h_j(\mathbf{S})\}_j$. Furthermore, different from data perturbations, the perturbed graph introduces errors at each GST layer.  This makes the analysis non-trivial.}. Although the stability of a transform is unquestionably important, pGST should be sensitive to  signal perturbation that have some specific pattern. For example, signal perturbations may be distributed over all the nodes but are localized in the graph spectrum.

Specifically, two classes of input signal perturbations are considered, each having distinct graph spectral properties. 

\noindent\textbf{Highly localized noise}.
Here the energy of the perturbation signal is localized in the graph spectrum; that is,
\begin{align}
\hspace*{-0.3cm}
\mathbf{E}_L^{(\epsilon)} \ = \ \{ \bm{\delta}_{L} \in \mathbb{R}^N | 
\widehat{\bm{\delta}}_{L}=\mathbf{V}\bm{\delta}_{L},
\left\| \widehat{\bm{\delta}}_{L} \right\| \leq \epsilon,
\left\| \widehat{\bm{\delta}}_{L} \right\|_0 = 1
\}.\label{eq:loc}
\end{align}

\noindent\textbf{Random noise}.
Here the energy of a perturbation signal is uniformly spread over the graph spectrum; that is,
\begin{align}
\mathbf{E}_R^{(\epsilon)} \ = \ \{ \bm{\delta}_{R} \in \mathbb{R}^N | 
\widehat{\bm{\delta}}_{R}=\mathbf{V}\bm{\delta}_{R},
[\widehat{\bm{\delta}}_{R}]_{n}^2 \le \epsilon/N
\}.\label{eq:rand}
\end{align}
The following corollary establishes the sensitivity bound for these classes of perturbation signals. 
\begin{corollary}
Let $\bm{\Psi}(\cdot)$ be a pGST with $L$ layers and ${J}=N$ filter. The $j$th filter has a transform function  $h_j(\mathbf{S})$ with $\widehat{{h}}_j(\lambda_n)=1$ for $n=j$, and 0 for $n\ne j$.  Consider two types of the perturbed inputs: $\tilde{\mathbf{x}}_L\define\mathbf{x}+{\bm{\delta}}_L$, where ${\bm{\delta}}_L \in \mathbf{E}_L$,
and 
$\tilde{\mathbf{x}}_R\define\mathbf{x}+{\bm{\delta}}_R$, where ${\bm{\delta}}_R \in \mathbf{E}_R$. It then holds that
\begin{eqnarray*}
\label{eq:assump2}
&& 
 \|\bm{\Psi}(\mathbf{x})-\bm{\Psi} (\tilde{\mathbf{x}}_R)\|\le
     \sqrt{
    \left(\sum_{\ell=1}^{L}\frac{F_{\ell}\framebound^{2\ell}}{N}+1\right)
    \epsilon}
\\
&&
    \|\bm{\Psi}(\mathbf{x})-\bm{\Psi} (\tilde{\mathbf{x}}_L)\|
    \le
\sqrt{
    \left(\sum_{\ell=2}^{L}F_{\ell}'\framebound^{2\ell}+2\right)
    \epsilon}
\end{eqnarray*}
where $F_{\ell}:=|\mathcal{P}^{(\ell)}\cup \mathcal{T}_R|$ and $F_{\ell}':=|\mathcal{P}^{(\ell)}\cup \mathcal{T}_L|$ are the retained features in the cases of random and localized noise. \end{corollary}
Corollary 1 suggests that pGST is more sensitive to $\bm{\delta}_L$ than to $\bm{\delta}_R$. This is appealing because pGST can detect transient changes in the graph spectral domain while ignoring small random perturbations. These results will be numerically corroborated in Sec.~\ref{sec:exp}. Specifically,  Fig.~\ref{fig:pGSTexnoise} shows that relative to the original $\mathcal{T}$, the pruned tree changes considerably when the noise is localized in the input spectrum, while the tree remains the same for random noise.

\section{Experiments}~\label{sec:exp}
\vspace{1mm}
This section evaluates the performance of our pGST in various graph classification tasks.  Graph classification amounts to predicting a label $y_i$ given $\mathbf{x}_i$ and $\mathbf{S}_i$ for the $i$th graph. Our pGST extracts the efficient and stable representation $\bm{\Psi}(\mathbf{x}_i)$, which is utilized as a feature vector for predicting $y_i$. During training, the structure of the pGST $\mathcal{T}$ is determined, which is kept fixed during validation and testing. The parameter $\tau$ is selected via cross-validation.  Our goal is to provide tangible answers to the following research questions. 
\vspace{1mm}
\begin{itemize}
    \vspace{1mm}\item[\textbf{RQ1}] How does the proposed pGST compare to GST?
    \vspace{1mm}\item[\textbf{RQ2}] How does pGST  compare to state-of-the-art GCNs in graph-based classification tasks? 
    \vspace{1mm}\item[\textbf{RQ3}] Given graph data, what are pruned scattering patterns?
    \vspace{1mm}\item[\textbf{RQ4}] What is the impact of signal perturbations in the scattering patters?
    \vspace{1mm}\item[\textbf{RQ5}] Can the effective GCN depth linked with the pGST depth?
\end{itemize}
\vspace{1mm}
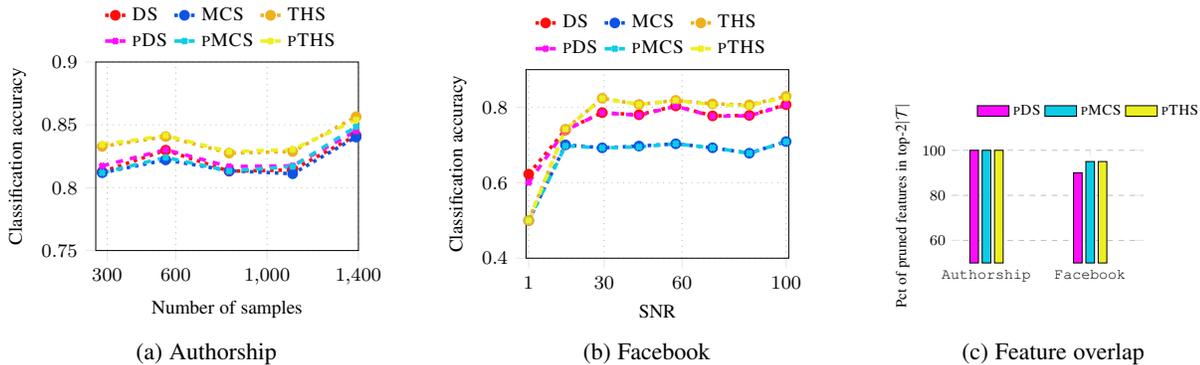
\begin{figure*}
  \hspace{-1.3cm}  \begin{subfigure}[b]{0.3\textwidth}{
\begin{tikzpicture}[scale=0.9, transform shape]

\begin{axis}[width=0.3\mywidth,
height=0.4\myheight,
at={(0\mywidth,0\myheight)},
legend entries={{\textsc{DS}},{\textsc{MCS}},{\textsc{THS}},{\textsc{pDS}},{\textsc{pMCS}},{\textsc{pTHS}}},
legend style={draw=white!80.0!black},
tick align=outside,
tick pos=left,
x grid style={white!69.01960784313725!black},
xlabel={Number of samples},
label style={font=\footnotesize},ticklabel style={font=\footnotesize},
xmin=250, xmax=1400,
xtick={300,600,1000,1400},
ylabel={Classification accuracy},
ymin=0.75, ymax=0.9,
legend columns=3,
xmajorgrids,
ymajorgrids,
grid style={dotted},
legend style={
	at={(0,1.015)}, 
	anchor=south west, legend cell align=left, align=left,
	draw=none
	, font=\footnotesize}]
\addplot [line width=\mylinewidth, dotted, DiffScater, mark=*, mark size=\markwidth, mark options={solid}]
table [row sep=\\]{%
278	0.81264182 \\
556	0.82987552 \\
834	0.81311239 \\
1112 0.81415094 \\
1390	0.84155844 \\
};
\addplot [line width=\mylinewidth, dotted,MonicCubic, mark=*, mark size=\markwidth, mark options={solid}]
table [row sep=\\]{%
278	0.81191248 \\
556	0.82209544 \\
834	0.81354467 \\
1112 0.81108491 \\
1390	0.84025974 \\
};
\addplot [line width=\mylinewidth,dotted, TightHann, mark=*, mark size=\markwidth, mark options={solid}]
table [row sep=\\]{%
278	0.8328201 \\
556	0.8406639 \\
834	0.82752161 \\
1112 0.82877358 \\
1390	0.85649351 \\
};
\addplot [dashed,line width=\mylinewidth, pDiffScater, mark=x, mark size=\markwidth, mark options={solid}]
table [row sep=\\]{%
278	0.81734522 \\
556	0.830252697 \\
834	0.8168012 \\
1112 0.8174434 \\
1390	0.845036364 \\
};
\addplot [dashed,line width=\mylinewidth, pMonicCubic, mark=x, mark size=\markwidth, mark options={solid}]
table [row sep=\\]{%
278	0.81183144 \\
556	0.82409544 \\
834	0.81325648 \\
1112 0.81684906 \\
1390	0.84855844 \\
};
\addplot [dashed,line width=\mylinewidth, pTightHann, mark=x, mark size=\markwidth, mark options={solid}]
table [row sep=\\]{%
278	0.83435981 \\
556	0.84149378 \\
834	0.82867435 \\
1112 0.83042453 \\
1390	0.8538961 \\
};
\end{axis}
\end{tikzpicture}}
    \caption{Authorship}
    \end{subfigure}%
    \hspace{.3cm}
    \begin{subfigure}[b]{0.3\textwidth}
  {
\begin{tikzpicture}[scale=0.9, transform shape]

\begin{axis}[width=0.3\mywidth,
height=0.4\myheight,
at={(0\mywidth,0\myheight)},
legend entries={{\textsc{DS}},{\textsc{MCS}},{\textsc{THS}},{\textsc{pDS}},{\textsc{pMCS}},{\textsc{pTHS}}},
legend style={draw=white!80.0!black},
tick align=outside,
tick pos=left,
x grid style={white!69.01960784313725!black},
xlabel={SNR},
xmin=0, xmax=101,
xtick={1, 30, 60,100},
ylabel={Classification accuracy},
ymin=0.4, ymax=0.9,ticklabel style={font=\footnotesize},
legend columns=3,
xmajorgrids,label style={font=\footnotesize},
ymajorgrids,
grid style={dotted},
legend style={
	at={(0,1.015)}, 
	anchor=south west, legend cell align=left, align=left,
	draw=none
	, font=\footnotesize}]
\addplot [line width=\mylinewidth, DiffScater,dotted, mark=*, mark size=\markwidth, mark options={solid}]
table [row sep=\\]{
1 0.622875\\
15.14285714 0.7403\\
29.28571429 0.786\\
43.42857143 0.78035\\
57.57142857 0.80385
\\
71.71428571 0.777375\\
85.85714286 0.7787\\
100 0.807125\\
};
\addplot [line width=\mylinewidth, MonicCubic,dotted, mark=*, mark size=\markwidth, mark options={solid}]
table [row sep=\\]{
1 0.49995\\
15.14285714 0.70045\\
29.28571429 0.69265\\
43.42857143 0.69735\\
57.57142857 0.703375
\\
71.71428571 0.6928\\
85.85714286 0.6784\\
100 0.70925 \\
};
\addplot [line width=\mylinewidth, TightHann,dotted, mark=*, mark size=\markwidth, mark options={solid}]
table [row sep=\\]{
1 0.49995\\
15.14285714 0.741875\\
29.28571429 0.824075\\
43.42857143 0.807725\\
57.57142857 0.81855
\\
71.71428571 0.809325\\
85.85714286 0.806725\\
100 0.829275\\
};
\addplot [dashed,line width=\mylinewidth, pDiffScater, mark=x, mark size=\markwidth, mark options={solid}]
table [row sep=\\]{
1 0.60075\\
15.14285714 0.739325\\
29.28571429 0.785925\\
43.42857143 0.778725\\
57.57142857 0.803125
\\
71.71428571 0.77645\\
85.85714286 0.777625\\
100 0.80565\\
};
\addplot [dashed,line width=\mylinewidth, pMonicCubic, mark=x, mark size=\markwidth, mark options={solid}]
table [row sep=\\]{
1 0.49995\\
15.14285714 0.697775\\
29.28571429 0.692525\\
43.42857143 0.6952\\
57.57142857 0.703375
\\
71.71428571 0.6928\\
85.85714286 0.6784\\
100 0.70885\\
};
\addplot [dashed,line width=\mylinewidth, pTightHann, mark=x, mark size=\markwidth, mark options={solid}]
table [row sep=\\]{
1 0.49995\\
15.14285714 0.7418\\
29.28571429 0.823825\\
43.42857143 0.80755\\
57.57142857 0.818525
\\
71.71428571 0.807\\
85.85714286 0.804625\\
100 0.827175\\
};
\end{axis}
\end{tikzpicture}}
    \caption{Facebook}
    \end{subfigure}
\hspace{0.3cm}
    \begin{subfigure}[b]{0.24\textwidth}
%
%
%
\begin{tikzpicture}[scale=0.7, transform shape]

\begin{axis}[%
axis line style={draw=none},
width=0.20\mywidth,
height=0.2\myheight,
at={(0\mywidth,0\myheight)},
scale only axis,
bar shift auto,legend entries={{\textsc{pDS}},{\textsc{pMCS}},{\textsc{pTHS}}},
log origin=infty,ticklabel style={font=\footnotesize},
xmin=0.1,
xmax=1.20909090909,
xtick={0.3,0.9,3.5,4.75,6,7.25,8.5},
xticklabels={{\texttt{Authorship}},{\texttt{Facebook}}},
ymin=50,
ymax=101,
yminorticks=true,
ylabel style={font=\color{white!15!black}},
ylabel={Pct of pruned features in top-2$|\mathcal{T}|$},
ylabel near ticks,
axis background/.style={fill=white},
ymajorgrids,label style={font=\footnotesize},
 legend columns=3,
legend style={at={(0.1,1.2)}, font=\footnotesize,anchor=south west, legend cell align=left, align=left, draw=none},grid style={dashed}
]

\addplot[ybar, bar width=0.05, fill=pDiffScater, area legend] table[row sep=crcr] {%
	0.3	 100\\
	0.9 90\\
};

\addplot[ybar, bar width=0.05,  fill=pMonicCubic, area legend] table[row sep=crcr] {%
	0.3	100 \\
	0.9 95\\
};

\addplot[ybar,bar width=0.05, fill=pTightHann, area legend] table[row sep=crcr] {%
0.3	100\\
0.9 95\\
};
\end{axis}
\end{tikzpicture}%
    \vspace{0.15cm}
    \caption{Feature overlap}
    \end{subfigure}
    \caption{Classification accuracy against number of samples in the authorship attribution (a) and SNR in dB for source localization (b).  The percentage of features after prunning retained in the top-$2|\mathcal{T}|$ most important GST features given by the SVM classifier (c).
    }
    \label{fig:sgam}
    \vspace{0mm}
\end{figure*}

\begin{figure}
%
%
%
\begin{tikzpicture}[scale=0.7, transform shape]

\begin{axis}[%
axis line style={draw=none},
width=0.2\mywidth,
height=0.2\myheight,
at={(0\mywidth,0\myheight)},
scale only axis,
bar shift auto,
log origin=infty,
xmin=0.709090909090909,
xmax=1.7590909090909,
xtick={1,1.5,3.5,4.75,6,7.25,8.5},
xticklabels={{\texttt{Authorship}},{\texttt{Facebook}}},
ymode=log,
ymin=0.001,
ymax=2,
yminorticks=true,
ytick={0.001,0.1},
xlabel={},
ylabel={Seconds},
ylabel near ticks,
axis background/.style={fill=white},
ymajorgrids,
legend entries={{\textsc{DS}},{\textsc{pDS}},{\textsc{MCS}},{\textsc{pMCS}},{\textsc{THS}},{\textsc{pTHS}}},
 legend columns=3,
label style={font=\footnotesize},ticklabel style={font=\footnotesize},
legend style={at={(0.0,1.)}, anchor=south west, font=\footnotesize, legend cell align=left, align=left, draw=none},grid style={dashed}
]
\addplot[ybar, bar width=0.05, fill=DiffScater, draw=black, area legend] table[row sep=crcr] {%
    1	  3.7389838695526123 \\
	1.5 0.08213615417480469 \\
};
\addplot[ybar, bar width=0.05, fill=pDiffScater, area legend] table[row sep=crcr] {%
	1	 0.034746646881103516\\
	1.5 0.0048694610595703125\\
};

\addplot[ybar, bar width=0.05, fill=MonicCubic, draw=black, area legend] table[row sep=crcr] {%
	1	1.8675472736358643  \\
	1.5 0.0725557804107666 \\
};
\addplot[ybar, bar width=0.05,  fill=pMonicCubic, area legend] table[row sep=crcr] {%
	1	0.21765422821044922 \\
	1.5 0.006527662277221689\\
};

\addplot[ybar, bar width=0.05, fill=TightHann, draw=black, area legend] table[row sep=crcr] {%
1 1.4091877937316895\\
1.5 0.031209945678710938\\
};
\addplot[ybar, bar width=0.05,  fill=pTightHann, area legend] table[row sep=crcr] {%
1	0.218569278717041\\
1.5 0.00223143272442\\
};
\end{axis}
\end{tikzpicture}%
    \caption{Runtime comparison of the scattering transforms for Fig.~\ref{fig:sgam}.}\label{fig:runtime}
    \end{figure}
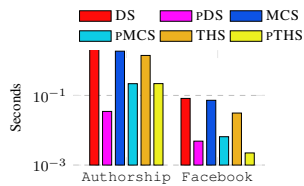


\subsection{Pruned GST compared to traditional GST}
To address RQ1, we reproduce the experiments of two tasks in \cite{gama2019stability}: authorship attribution and source localization. For the scattering transforms, we consider three  implementations of graph filter banks: the diffusion wavelets (\textsc{DS}) in \cite{gama2018diffusion}, the monic cubic  wavelets (\textsc{MCS}) in \cite{hammond2011wavelets} and the tight Hann wavelets (\textsc{THS}) in \cite{shuman2015spectrum}.\footnote{\textsc{pDS}, \textsc{pMCS}, \textsc{pTHS} denote the pruned versions of these transforms.} The scattering transforms use $J=5$ filters, $L=5$ layers, and $\tau=0.01$. The extracted features from GSTs are subsequently utilized by a linear support vector machine (SVM) classifier. 

Authorship attribution amounts to determining if a certain text was written by a specific author. Each text is represented by a graph with $N=244$, where words (nodes) are connected based on their relative positions in the text, and $\mathbf{x}$ is a bag-of-words representation of the text; see also~\cite{gama2018diffusion}.  Fig. \ref{fig:sgam} (a) reports the classification accuracy as the number of training samples (texts) increases. GSTs utilize $\sum_{\ell=1}^5{5}^\ell=781$ scattering coefficients, while pGSTs rely only on  $|\mathcal{T}|=61$ for \textsc{pDS}, $|\mathcal{T}|=30$ for \textsc{pMCS}, and $|\mathcal{T}|=80$ for \textsc{pTHS}. Evidently, pGST achieves comparable performance as the baseline GST, whereas pGST uses only a subset of features ($12.8\%, 3.8\%$ and $10.2\%$, respectively).  The SVM classifier provides a coefficient that weighs each scattering scalar. The magnitude of each coefficient shows the importance of the corresponding scattering feature in the classification. Fig.~\ref{fig:sgam} (c) depicts the percentage of features retained after prunning  in the top-$2|\mathcal{T}|$ most important GST features given by the SVM classifier.  It is observed, that although pGST does not take into account the labels, the retained features are indeed informative for classification.

Source localization amounts to recovering the source of a rumor given a diffused signal over a Facebook subnetwork with $N=234$; see the detailed settings in \cite{gama2018diffusion}. Fig. \ref{fig:sgam} (b) shows the classification accuracy of the scattering transforms as the SNR (in dB) increases. In accordance to Lemma 1 and Theorem 2, both pGST and GST are stable over a wide range of SNR. Furthermore, the performance of pGST matches that of GST, while the pGST uses only a subset of features. Finally, Fig. \ref{fig:runtime} depicts the runtime of the different scattering approaches, where the computational advantage of the pruned methods is evident. 
\subsubsection{Ablation study}
\begin{figure*}
\begin{subfigure}[b]{0.3\textwidth}
\begin{tikzpicture}

\begin{axis}[width=0.3\mywidthbb,
height=0.4\myheightbb,
at={(0\mywidth,0\myheight)},
legend entries={\textsc{pDS},\textsc{pMCS},\textsc{pTHS}},
legend style={draw=white!80.0!black},
tick align=outside,
tick pos=left,
x grid style={white!69.01960784313725!black},
xlabel={$\tau$},ticklabel style={font=\footnotesize},
label style={font=\footnotesize},
xmin=0.00001, xmax=1,
xmode=log,
ylabel={Classification accuracy},
ymin=0.50, ymax=0.98,
legend columns=6,
xmajorgrids,
ymajorgrids,
grid style={dotted},
legend style={
	at={(0,1.015)}, 
	anchor=south west, legend cell align=left, align=left,
	draw=none
	, font=\footnotesize}]
\addplot [dashed,line width=\mylinewidth, pDiffScater, mark=x, mark size=\markwidth, mark options={solid}]
table [row sep=\\]{
1.00000000e-05 0.8475\\
5.17947468e-05 0.8295\\ 
2.68269580e-04 0.795\\ 
1.38949549e-03  0.795\\
 7.19685673e-03 0.786525\\
 3.72759372e-02 0.787  \\
 1.93069773e-01 0.77\\
 1.00000000e+00 0.69625\\
};
\addplot [dashed,line width=\mylinewidth, pMonicCubic, mark=x, mark size=\markwidth, mark options={solid}]
table [row sep=\\]{%
1.00000000e-05 0.78875\\
5.17947468e-05 0.71475\\ 
2.68269580e-04 0.70625\\ 
1.38949549e-03  0.68475\\
 7.19685673e-03 0.69075\\
 3.72759372e-02 0.699  \\
 1.93069773e-01 0.5775\\
 1.00000000e+00 0.566\\
};
\addplot [dashed,line width=\mylinewidth, pTightHann, mark=x, mark size=\markwidth, mark options={solid}]
table [row sep=\\]{
1.00000000e-05 0.902\\
5.17947468e-05 0.937\\ 
2.68269580e-04 0.943\\ 
1.38949549e-03  0.95\\
 7.19685673e-03 0.93\\
 3.72759372e-02 0.93  \\
 1.93069773e-01 0.93\\
 1.00000000e+00 0.90625\\
};
\end{axis}
\end{tikzpicture}
\caption{Classification accuracy}
\end{subfigure}~
\begin{subfigure}[b]{0.3\textwidth}
\begin{tikzpicture}

\begin{axis}[width=0.3\mywidthbb,
height=0.4\myheightbb,
at={(0\mywidth,0\myheight)},
legend entries={\textsc{pDS},\textsc{pMCS},\textsc{pTHS}},legend style={draw=white!80.0!black},
tick align=outside,
tick pos=left,ticklabel style={font=\footnotesize},
x grid style={white!69.01960784313725!black},
label style={font=\footnotesize},
xmin=0.00001, xmax=1,
xmode=log,
xlabel=$\tau$,
ylabel={$|\mathcal{T}|$},
ymin=0, ymax=10000, ymode=log,
legend columns=6,
xmajorgrids,
ymajorgrids,
grid style={dotted},
legend style={
	at={(0,1.015)}, 
	anchor=south west, legend cell align=left, align=left,
	draw=none
	, font=\footnotesize}]
\addplot [dashed,line width=\mylinewidth, pDiffScater, mark=x, mark size=\markwidth, mark options={solid}]
table [row sep=\\]{
1.00000000e-05 9331\\
5.17947468e-05 9331\\ 
2.68269580e-04 9124\\ 
1.38949549e-03  2781\\
 7.19685673e-03 364\\
 3.72759372e-02 63  \\
 1.93069773e-01 6\\
 1.00000000e+00 1\\
};
\addplot [dashed,line width=\mylinewidth, pMonicCubic, mark=x, mark size=\markwidth, mark options={solid}]
table [row sep=\\]{
1.00000000e-05 7528\\
5.17947468e-05 6439\\ 
2.68269580e-04 5516\\ 
1.38949549e-03  3681\\
 7.19685673e-03 610\\
 3.72759372e-02 120  \\
 1.93069773e-01 12\\
 1.00000000e+00 1\\
};
\addplot [dashed,line width=\mylinewidth, pTightHann, mark=x, mark size=\markwidth, mark options={solid}]
table [row sep=\\]{
1.00000000e-05 3895\\
5.17947468e-05 3774\\ 
2.68269580e-04 3111\\ 
1.38949549e-03  2838\\
 7.19685673e-03 2034\\
 3.72759372e-02 814  \\
 1.93069773e-01 78\\
 1.00000000e+00 6\\
};
\end{axis}
\end{tikzpicture}
\caption{Number of active features $|\mathcal{T}|$}
\end{subfigure}~\begin{subfigure}[b]{0.3\textwidth}\centering
\begin{tikzpicture}

\begin{axis}[width=0.3\mywidthbb,
height=0.4\myheightbb,
at={(0\mywidth,0\myheight)},
legend entries={\textsc{pDS},\textsc{pMCS},\textsc{pTHS}},legend style={draw=white!80.0!black},
tick align=outside,
tick pos=left,ticklabel style={font=\footnotesize},
x grid style={white!69.01960784313725!black},
xlabel={$\tau$},
label style={font=\footnotesize},
xmin=0.00001, xmax=1,
xmode=log,
ylabel={Seconds},
ymin=0.0001, ymax=1,ymode=log,
legend columns=6,
xmajorgrids,
ymajorgrids,
grid style={dotted},
legend style={
at={(0,1.015)}, 
	anchor=south west, legend cell align=left, align=left,
	draw=none
	, font=\footnotesize}]
\addplot [dashed,line width=\mylinewidth, pDiffScater, mark=x, mark size=\markwidth, mark options={solid}]
table [row sep=\\]{
1.00000000e-05 0.4771544933319092\\
5.17947468e-05 0.34178805351257324\\ 
2.68269580e-04 0.33472187042236328\\ 
1.38949549e-03  0.06737\\
 7.19685673e-03 0.0144\\
 3.72759372e-02 0.00226164  \\
 1.93069773e-01 0.00145841\\
 1.00000000e+00 0.00195909\\
};
\addplot [dashed,line width=\mylinewidth, pMonicCubic, mark=x, mark size=\markwidth, mark options={solid}]
table [row sep=\\]{
1.00000000e-05 0.3677945137023926\\
5.17947468e-05 0.2249\\ 
2.68269580e-04 0.2134\\ 
1.38949549e-03  0.22939\\
 7.19685673e-03 0.04277\\
 3.72759372e-02 0.01866  \\
 1.93069773e-01 0.00182271\\
 1.00000000e+00 0.00152731\\
};
\addplot [dashed,line width=\mylinewidth, pTightHann, mark=x, mark size=\markwidth, mark options={solid}]
table [row sep=\\]{
1.00000000e-05 0.13809704780578613\\
5.17947468e-05 0.14978241920471191\\ 
2.68269580e-04 0.115568\\ 
1.38949549e-03  0.1143\\
 7.19685673e-03 0.1307693\\
 3.72759372e-02 0.0538533  \\
 1.93069773e-01 0.00498533\\
 1.00000000e+00 0.0022614\\
};
\end{axis}
\end{tikzpicture}\caption{Runtime in seconds}\end{subfigure}
    \caption{Performance of pGSTs for varying $\tau$.}
    \label{fig:senstau}
\end{figure*}
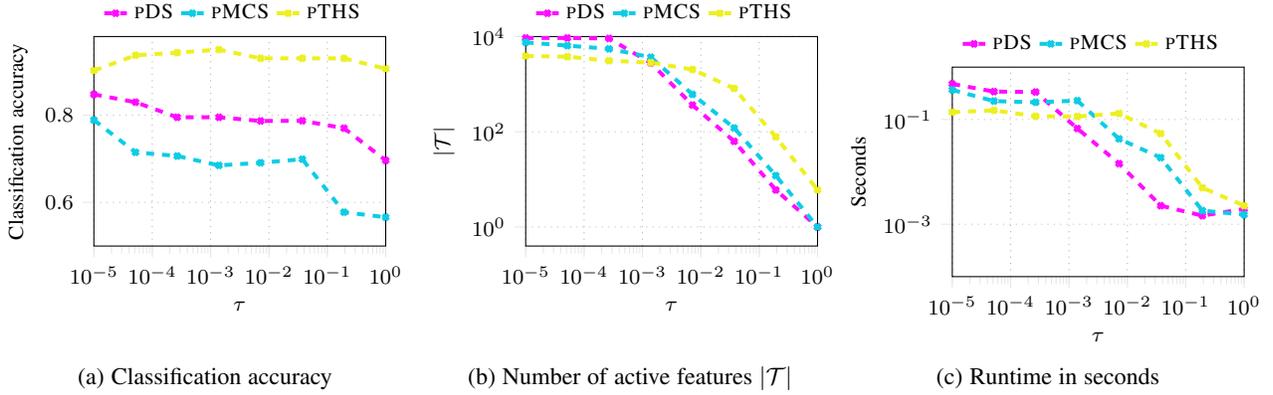

\begin{figure}
\begin{tikzpicture}

\begin{axis}[width=0.3\mywidthbb,
height=0.4\myheightbb,
at={(0\mywidth,0\myheight)},
legend entries={\textsc{pDS},\textsc{pMCS},\textsc{pTHS}},
legend style={draw=white!80.0!black},
tick align=outside,
tick pos=left,
x grid style={white!69.01960784313725!black},
xlabel={$L$},ticklabel style={font=\footnotesize},
label style={font=\footnotesize},
xmin=1, xmax=6,
ylabel={Classification accuracy},
ymin=0.60, ymax=1,
legend columns=6,
xmajorgrids,
ymajorgrids,
grid style={dotted},
legend style={
	at={(0,1.015)}, 
	anchor=south west, legend cell align=left, align=left,
	draw=none
	, font=\footnotesize}]
\addplot [dashed,line width=\mylinewidth, pDiffScater, mark=x, mark size=\markwidth, mark options={solid}]
table [row sep=\\]{
1 0.7453\\
2 0.785\\
3 0.8323\\
4 0.8523\\
5 0.8515\\
6 0.8542\\
};
\addplot [dashed,line width=\mylinewidth, pMonicCubic, mark=x, mark size=\markwidth, mark options={solid}]
table [row sep=\\]{
1 0.6054\\
2 0.6758\\
3 0.671\\
4 0.693\\
5 0.7346\\
6 0.7663\\
};
\addplot [dashed,line width=\mylinewidth, pTightHann, mark=x, mark size=\markwidth, mark options={solid}]
table [row sep=\\]{
1 0.6354\\
2 0.7032\\
3 0.8729\\
4 0.8759\\
5 0.9744\\
6 0.9712\\
};
\end{axis}
\end{tikzpicture}
\caption{Classification accuracy over $L$}\label{fig:sensL}
\end{figure}
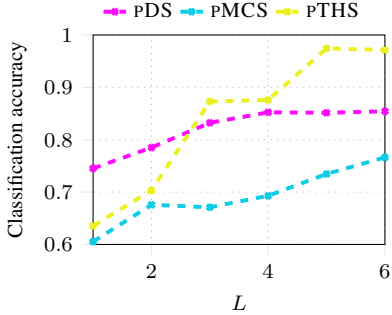
\begin{figure}
\begin{tikzpicture}

\begin{axis}[width=0.3\mywidthbb,
height=0.4\myheightbb,
at={(0\mywidth,0\myheight)},
legend entries={\textsc{pDS},\textsc{pMCS},\textsc{pTHS}},
legend style={draw=white!80.0!black},
tick align=outside,
tick pos=left,
x grid style={white!69.01960784313725!black},
xlabel={$J$},ticklabel style={font=\footnotesize},
label style={font=\footnotesize},
xmin=3, xmax=10,
ylabel={Classification accuracy},
ymin=0.60, ymax=0.9,
legend columns=6,
xmajorgrids,
ymajorgrids,
grid style={dotted},
legend style={
	at={(0,1.015)},
	anchor=south west, legend cell align=left, align=left,
	draw=none
	, font=\footnotesize}]
\addplot [dashed,line width=\mylinewidth, pDiffScater, mark=x, mark size=\markwidth, mark options={solid}]
table [row sep=\\]{%
3 0.7851\\
4 0.7737\\
5 0.7995\\
6 0.8232\\
7 0.8388\\
8 0.8274\\
9 0.8302\\
10 0.827\\
};
\addplot [dashed,line width=\mylinewidth, pMonicCubic, mark=x, mark size=\markwidth, mark options={solid}]
table [row sep=\\]{%
3 0.6582\\
4 0.6839\\
5 0.6637\\
6 0.7317\\
7 0.7541\\
8 0.7623\\
9 0.8045\\
10 0.829\\
};
\addplot [dashed,line width=\mylinewidth, pTightHann, mark=x, mark size=\markwidth, mark options={solid}]
table [row sep=\\]{%
3 0.7977\\
4 0.796\\
5 0.8291\\
6 0.8205\\
7 0.8219\\
8 0.8354\\
9 0.8342\\
10 0.8313\\
};
\end{axis}
\end{tikzpicture}
\caption{Classification accuracy over $J$}\label{fig:sensJ}
\end{figure}
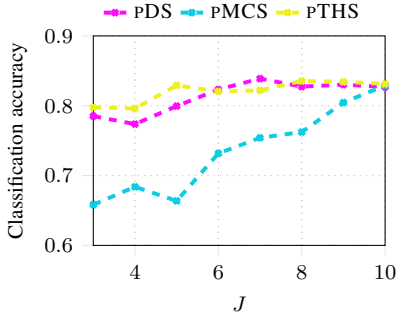

Fig.~\ref{fig:senstau} reports how the pGST is affected by varying the threshold $\tau$ in the task of source localization, with $J=6$ and $L=5$. Fig.~\ref{fig:senstau} (a) shows the classification accuracy that generally decreases as $\tau$ increases since the number of active features $|\mathcal{T}|$ decreases; cf. Fig.~\ref{fig:senstau} (b). Fig.~\ref{fig:senstau} (c) reports the runtime in seconds. Fig.~\ref{fig:sensL},~\ref{fig:sensJ} showcase the classification performance of pGST with $\tau=0.01$ for varying $L$ with $J=3$, and for varying $J$ with $L=3$, respectively. It is observed, that the classification performance generally improves with $L$ and $J$.

\begin{table}[t]
		\hspace{0cm}
		\centering
		\caption{Dataset characteristics}
		\begin{tabular}{c c c c c}
			\hline
			\textbf {Dataset} &  {Graphs}  &  {Features} $F$  & {Max $N$ per graph} \\
			\hline\hline
			Collab  &  5000& 1 & 492\\\hline
			 D\&D& 1178 & 89 & 5748\\\hline
			Enzymes  & 600 & 3 & 126\\\hline
			Proteins &1113 & 3 &620
		\end{tabular}
		\vspace{0.0cm}
		\label{tab:dataset}
\end{table}
\begin{table}[]
\centering 
\caption{Graph classification accuracy.}
\resizebox{0.8\textwidth}{!}{ 
{
\newcolumntype{g}{>{\columncolor{gray}}c}
\begin{tabular}{@{}clccHccc@{}}\cmidrule[\heavyrulewidth]{2-7}
& \multirow{3}{*}{\vspace*{8pt}\textbf{Method}}&\multicolumn{4}{c}{\textbf{Data Set}}\\\cmidrule{3-7}
& & {\textsc{Enzymes}} & {\textsc{D\&D}} & {\textsc{Reddit-Multi-12k}} & {\textsc{Collab}} & {\textsc{Proteins}} 
\\ \cmidrule{2-7}\cmidrule{2-7}
\multirow{2}{*}{\rotatebox{90}{\hspace*{-6pt}Kernel}} 
& \textsc{Shortest-path} & 42.32 & 78.86 & 36.93 & 59.10  & 76.43   \\   
\cmidrule{2-7}
& \text{WL-OA} & 60.13  & 79.04	 & 44.38  & 80.74  & 75.26    \\       \cmidrule{2-7}
\cmidrule{2-7}
& \textsc{PatchySan} & -- & 76.27	 & 41.32   & 72.60 &  75.00   \\ \cmidrule{2-7}
\multirow{7}{*}{\rotatebox{90}{GNNs}} 
& \textsc{GraphSage} &  54.25 & 75.42 	 & 42.24  & 68.25  & 70.48\\ \cmidrule{2-7}
& \textsc{ECC}  & 53.50  & 74.10 & 41.73  & 67.79 &   72.65   \\	\cmidrule{2-7}
& \textsc{Set2set} &  60.15 & 78.12  & 43.49 & 71.75 & 74.29  \\ \cmidrule{2-7}
& \textsc{SortPool} & 57.12 & 79.37  & 41.82 & 73.76  & 75.54   \\     \cmidrule{2-7}
& \textsc{DiffPool-Det} & 58.33 & 75.47 & 46.18 & \textbf{82.13} & 75.62 \\ \cmidrule{2-7}
& \textsc{DiffPool-NoLP} & 62.67  & 79.98	 & 46.44  & 75.63   &  77.42   \\ \cmidrule{2-7}
& \textsc{DiffPool} & \textbf{64.23}  & {81.15}	 & \textbf{47.04}  & 75.50   &  {78.10} \\\cmidrule{2-7}
 \cmidrule{2-7}
\multirow{3}{*}{\rotatebox{90}{Scattering}}
&\textsc{GSC} &53.88 & 76.57& & 76.88 &74.03\\ \cmidrule{2-7}
&\textsc{GST} &59.84 &79.28 &  &  77.32& 76.23\\\cmidrule{2-7} \cmidrule{2-7}
& \textsc{pGST} (Ours) & {60.25}  & \textbf{81.27}	 &  & 78.40   &  \textbf{78.57}\\
\cmidrule[\heavyrulewidth]{2-7}
\end{tabular}}}
\label{tab:results}
\end{table}

\subsection{Pruned GST for graph-based classification tasks}
In response to RQ2, we consider three graph-based classification tasks: graph classification, graph-based 3D point cloud classification, and semi-supervised node classification.

\subsubsection{Graph classification} We compare the proposed pGST with the following state-of-the-art methods\footnote{For the competing approaches we report the 10-fold cross-validation numbers reported by the original authors; see also \cite{ying2018hierarchical}.}. The kernel methods shortest-path~\cite{borgwardt2005shortest}, and Weisfeiler-Lehman optimal assignment (WL-OA)~\cite{kriege2016valid}; the deep learning approaches PatchySan~\cite{niepert2016learning}, GraphSage~\cite{hamilton2017representation}, edge-conditioned filters in CCNs (ECC)~\cite{simonovsky2017dynamic}, Set2Set~\cite{vinyals2015order}, SortPool~\cite{zhang2018end}, and DiffPool~\cite{ying2018hierarchical}; and the geometric scattering classifier (GSC)~\cite{gao2019geometric}. Results are presented on the protein data sets D\&D, Enzymes and Proteins, and the scientific collaboration data set Collab. The parameters of these datasets are listed in Table~\ref{tab:dataset}. Since the Collab dataset did not have nodal features,  $\mathbf{x}$ was selected as the vector that holds the node degrees. We performed 10-fold cross validation and report the classification accuracy averaged over the 10 folds. The gradient boosting classifier is employed for pGST and GST with parameters chosen based on the performance on the validation set. The graph scattering transforms use the MC wavelet with $L=5$, $J=5$, and  $\tau=0.01$. Table~\ref{tab:results} lists the classification accuracy of the proposed and competing approaches.  Even without any training on the feature extraction step, the performance of pGST is comparable to the state-of-the-art deep supervised learning approaches across all datasets. GST and pGST outperform also GSC, since the latter uses a linear SVM to classify the scattering features.

\begin{figure*}
\begin{subfigure}[b]{0.45\textwidth}
\centering
\begin{tikzpicture}

\definecolor{color0}{rgb}{1,0.647058823529412,0}
\definecolor{color1}{rgb}{1,1,0}
\definecolor{color2}{rgb}{0.501960784313725,0,0.501960784313725}

\begin{axis}[width=0.45\mywidths,
height=0.55\myheights,
at={(0\mywidth,0\myheight)},
legend entries={\textsc{MCS},\textsc{pMCS},
\textsc{Pointnet}, \textsc{Pointnet++}, \textsc{3DShapeNets},\textsc{VoxNet}},
legend style={draw=white!80.0!black},
tick align=outside,
tick pos=left,
x grid style={white!69.01960784313725!black},
xlabel={Network depth $L$},
xmin=1.5, xmax=6.5,
ylabel={Classification accuracy},
ymin=0.7, ymax=0.95,
xmajorgrids,
ymajorgrids,
grid style={dotted},
legend columns=3,label style={font=\footnotesize},ticklabel style={font=\footnotesize},
legend style={
	at={(-0.1,1.015)}, 
	anchor=south west, legend cell align=left, align=left, draw=none
	, font=\footnotesize}]

\addplot [mark size=0.5, line width=\mylinewidth,MonicCubic,dotted, 
mark=*,  mark options={solid}]
table [row sep=\\]{%
2   0.7912155592\\		
3	0.818800648 \\
4	0.8243111831\\
5  0.8480551053\\
};
\addplot [mark size=0.5, line width=\mylinewidth, pMonicCubic,dashed, 
mark=x,  mark options={solid}]
table [row sep=\\]{%
2   0.7912155592\\		
3	0.828800648 \\
4	0.8343111831\\
5  0.8550551053\\
6 0.8496758509\\
};

\addplot [mark size=0.5, line width=\mylinewidth, pointnet, mark=square,  mark options={solid}]
table [row sep=\\]{%
2   0.87			\\		
3	0.87 \\
4	0.87 \\
5	0.87 \\
6 0.87\\
};
\addplot [mark size=0.5, line width=\mylinewidth, pointnetplus, mark=o,  mark options={solid}]
table [row sep=\\]{%
2   0.919			\\		
3	0.919 \\
4	0.919\\
5	0.919\\
6	0.919\\
};
\addplot [mark size=0.5, line width=\mylinewidth, 3dshapenets, mark=x,  mark options={solid}]
table [row sep=\\]{%
2   0.847			\\		
3	0.847 \\
4	0.847\\
5	0.847\\
6	0.847\\
};
\addplot [mark size=0.5, line width=\mylinewidth, voxnet, mark=*,  mark options={solid}]
table [row sep=\\]{%
2   0.859			\\		
3	0.859 \\
4	0.859\\
5	0.859\\
6	0.859\\
};

\end{axis}
\end{tikzpicture}
\caption{9843 clouds for training and 2468 for testing} 
\end{subfigure}
\begin{subfigure}[b]{0.45\textwidth}
\centering
\begin{tikzpicture}

\definecolor{color0}{rgb}{1,0.647058823529412,0}
\definecolor{color1}{rgb}{1,1,0}
\definecolor{color2}{rgb}{0.501960784313725,0,0.501960784313725}

\begin{axis}[width=0.45\mywidths,
height=0.55\myheights,
at={(0\mywidth,0\myheight)},
legend entries={\textsc{MCS},\textsc{pMCS}, \textsc{Pointnet++}},
legend style={draw=white!80.0!black},
tick align=outside,
tick pos=left,
x grid style={white!69.01960784313725!black},
xlabel={Network depth $L$},
xmin=1.5, xmax=5.5,
ylabel={Classification accuracy},
ymin=0.6, ymax=0.758165821954202,
xmajorgrids,
ymajorgrids,
grid style={dotted},
legend columns=3,label style={font=\footnotesize},ticklabel style={font=\footnotesize},
legend style={
	at={(0,1.015)}, 
	anchor=south west, legend cell align=left, align=left, draw=none
	, font=\footnotesize}]


\addplot [mark size=0.5, line width=\mylinewidth, MonicCubic,dotted, 
mark=*,  mark options={solid}]
table [row sep=\\]{%
2   0.6130957874\\		
3	0.7059728275 \\
4	0.6960608391\\
5  0.6844475775\\
};
\addplot [mark size=0.5, line width=\mylinewidth, pMonicCubic, dashed, 
mark=x,  mark options={solid}]
table [row sep=\\]{%
2   0.6130957874\\		
3	0.7059728275 \\
4	0.71051\\
5  0.710125\\
};

\addplot  [mark size=0.5, line width=\mylinewidth, pointnetplus, mark=o,  mark options={solid}]
table [row sep=\\]{%
2   0.685			\\		
3	0.685 \\
4	0.685 \\
5	0.685 \\
};

\end{axis}
\end{tikzpicture}
\caption{615 clouds for training and 11703 for testing}
\end{subfigure}
\caption{3D point cloud classification.}\label{fig:pcacc}
\end{figure*}
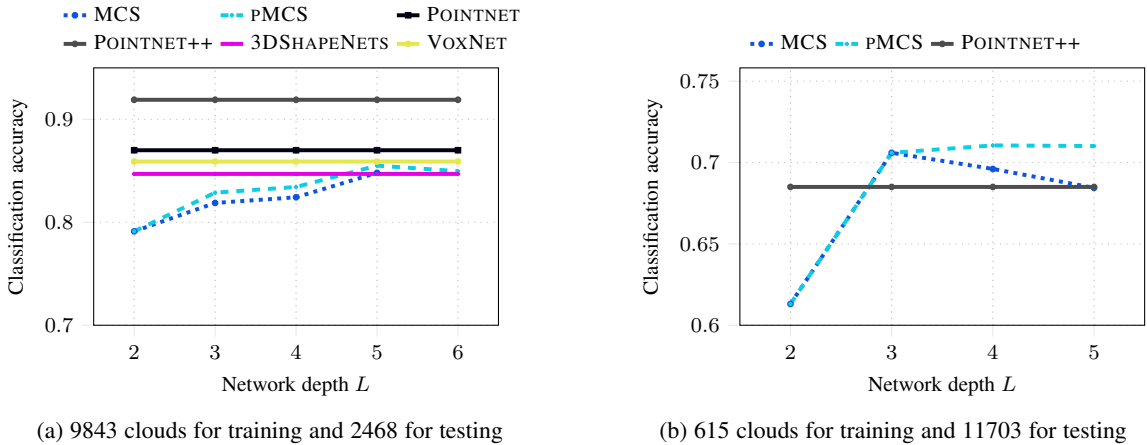

\subsubsection{Point cloud classification} 

We further test pGST in classifying 3D point clouds. Given a point cloud, a graph can be created by connecting points (nodes) to their nearest neighbors based on their Euclidian distance. Each node is also associated with 6 scalars denoting its x-y-z coordinates and RGB colors. For this experiment, GSTs are compared against PointNet++~\cite{qi2017pointnet,qi2017pointnet++}, 3dShapeNets~\cite{wu20153d} and VoxNet~\cite{maturana2015voxnet}, that are state-of-the-art deep learning approaches.  Fig.~\ref{fig:pcacc} reports the classification accuracy for the ModelNet40 dataset~\cite{wu20153d} for increasing $L$. In Fig.~\ref{fig:pcacc} (a) 9,843 clouds are used for training and 2,468 for testing using the gradient boosting classifier; whereas, in Fig.~\ref{fig:pcacc} (b) only 615 clouds are used for training and the rest for testing using a fully connected neural network classifier with 3 layers. The scattering transforms use an MC wavelet with $J=5$ for Fig.~\ref{fig:pcacc} (a), and  $J=9$ for Fig.~\ref{fig:pcacc} (b). Fig.~\ref{fig:pcacc} showcases that scattering transforms are competitive to state-of-the-art approaches, while pGST outperforms GST. This may be attributed to overfitting effects, since a large number of GST features is not informative. Furthermore, the exponential complexity of GSTs prevents their application with $L=6$. Fig.~\ref{fig:pcacc} (b) shows that when the training data are scarce, GST and pGST outperform the PointNet++, which requires a large number of training data to optimize over the network parameters.

\subsubsection{Semi-supervised node classification}{\label{sec:ssl}}
A task of major importance at the intersection of machine learning and 
network science is semi-supervised learning (SSL) over graphs. Given the topology ${\mathbf{S}}$, the features in the $N\times F$ matrix $\mathbf{X}$, and labels only at a subset $\mathcal{L}$ of nodes $\{y_{n}\}_{n\in\mathcal{L}}$ with $\mathcal{L} \subset\mathcal{V}$, the goal of is to predict the labels $\{y_{n}\}_{n\in\mathcal{U}}$ of the unlabeled set of nodes.
For this task, we utilize the pGST algorithm to extract the feature vectors $\{\mathbf{z}_{(p)}\}_{p\in \mathcal{T}}$ from $\mathbf{X}$, which are then utilized by a fully connected neural network with 3 layers to predict the missing labels. To facilitate comparison, we reproduce the experimental setup of \cite{kipf2016semi}, namely the same split of the data for training, validation, and testing sets, and compare the various methods in the Cora dataset with $N=2,708$ nodes, $C=7$ classes and $F=1,433$ features. The pGST algorithm employs the TH wavelet with $L=5$, $J=3$, and $\tau=0.01$. We also compare to~\cite{zou2019graph}, where the GST is employed, the full scattering coefficients are extracted and then dimensionality reduction is effected using PCA to handle the large number of features. The so obtained features are processed by a fully connected network. Table~\ref{tab:sslresult} reports the classification accuracy of several state-of-the-art methods. Our method that performs unsupervised feature extraction is highly competitive and outperforms all except the GAT approach. It is worth noting that pGST performs comparably to the GST approach in~\cite{zou2019graph}, which suggests that our spectrum-inspired criterion has a `PCA-like' effect. However, PCA needs all the scattering features to process them offline, while pGST prunes the scattering features on-the-fly, and allows for increased scalability.

\begin{table}[]
    \centering
    \caption{SSL classification accuracy.}
    \resizebox{0.32\textwidth}{!}{
    \begin{tabular}{c c}
    \hline
    \textbf{Method}  & \textbf{Cora}\\
    \hline
    ManiReg~\cite{belkin2006manifold} & 59.5  \\
    \hline
    SemiEmb~\cite{weston2012deep} &  59.0  \\
        \hline
    LP~\cite{zhu2003semi} & 68.0  \\
        \hline
    Planetoid~\cite{yang2016revisiting} &  75.7 \\
        \hline
    GCN~\cite{kipf2016semi} & 81.5 \\
    \hline
    GAT~\cite{velivckovic2017graph} &  \textbf{83.0} \\
    \hline
        GST~\cite{zou2019graph}&  81.9  \\
        \hline
    pGST (ours)& 81.9 \\
        \hline
    \end{tabular}
    }
    \label{tab:sslresult}
\end{table}

\subsection{Scattering patterns analysis}

To address RQ3, we depict the scattering structures of pGSTs, with an MC wavelet, $J=3$, and $L=5$, for the Collab, Proteins, and ModelNet40 datasets in Fig.~\ref{fig:pGSTex}. Evidently, graph data from various domains require an adaptive scattering architecture. Specifically, most tree nodes for the academic collaboration dataset are pruned, and hence most informative features reside at the shallow layers. This is consistent with the study in~\cite{wu2019simplifying}, which experimentally shows that deeper GCNs do not contribute as much for social network data. These findings are further supported by the small-world phenomenon in social networks, which suggests that the diameter of social networks is small~\cite{watts1998collective}.  On the other hand, the tree nodes for a 3D point cloud are minimally pruned, which is in line with the work in~\cite{li2019can} that showcases the advantage of deep GCNs in 3D point clouds classification. 

\subsubsection{Sensitivity of the pGST tree to signal perturbations}
Next, we evaluate noise effects in the pruning algorithm to provide an answer to RQ4. Corollary 1 suggests that localized noise will have a more prominent effect on the pGST than random noise. This is appealing because pGST can detect transient changes in the graph spectral domain, while ignoring small random perturbations. Here, we attempt to establish the changes to the \emph{structure} in the scattering patterns of pGST under localized and random noise.   

For this experiment, the Cora dataset is tested with  $N=2,708$ nodes, $C=7$ classes and $F=1,433$ features. We draw random and localized noise from~\eqref{eq:rand} and~\eqref{eq:loc} respectively, which is added to the nodal features. The noise is selected so that the signal to noise (SNR) energy ratio is -20dB. We observe that for smaller noise power, the scattering patterns are invariant, which corroborates the theoretical result in Lemma 2.

Fig.~\ref{fig:pGSTexnoise} plots the scattering patterns of the pGST using the TH wavelet with $L=3$, $J=5$, and $\tau=0.1$. It is observed that the set of pruned features given the original features without noise in Fig.~\ref{fig:orig}, is almost the same as the set of pruned features  given the randomly perturbed features in Fig~\ref{fig:ran}. On the other hand, the set of pruned features when the localized noise is employed, is quite different; see Fig.~\ref{fig:loc}.

\subsubsection{Sensitivity of the pGST tree to structural perturbations}
\label{sec:structural_sensitivity}
For the same setting as Fig.~\ref{fig:pGSTexnoise} the effect of the structural perturbations in the structure of pGST is also tested. The adjacency matrix was perturbed by adding noise in the eigenvalues such that the SNR is -20dB. Random noise in this scenario means that the noise is distributed over the eigenvalues, whereas localized noise affects only a subset of the eigenvalues significantly.  Fig.~\ref{fig:pGSTexnoise_structure} shows that the localized and random noise do not alter the scattering patterns significantly. This is in accordance with the stability results in Theorem 3.

\begin{figure*}
\begin{subfigure}{0.3\textwidth}
\centering{\scriptsize \begin{tikzpicture}[x = 1*\unit, y = 1*\unit,scale=\scaletrees, every node/.style={scale=1}, line/.style={
      draw,scale=1,
      -latex',
      shorten >=1pt
    },>=latex]
    
	\node at (0, 0) (o) {};

	\path (o) 
		      node (0) [fill = \colorNode,mnode, draw = \colorEdge] {};
    \foreach \dx\dy\nodeid in {{-2/-1/1},{-1/-1/2},{1/-1/4}}
         {\path (0) ++ (\dx*0.8*\mdistNodesLayerOne,\dy*\mdistBetweenLayers)
              node (\nodeid)[fill = \colorNode,mnode, draw = \colorEdge] 
              {};
        \path (0.south) 
       edge [draw = \colorGraphFilter, line width = \mmyArrowWidth] 
       (\nodeid.north east);}

    \foreach \pathid\dx\dy\nodeid in {{1/-2/-1/1},{1/-1/-1/2},{1/1/-1/4},
    {2/-2/-1/1},{2/-1/-1/2},
    {4/-2/-1/1},{4/-1/-1/2}
    }
         {\path (\pathid) ++ (\dx*0.6*\mdistNodesLayerTwo,\dy*\mdistBetweenLayers)
              node (\pathid\nodeid)[fill = \colorNode,mnode, draw = \colorEdge] 
              {};
              \path (\pathid.south) 
       edge [ draw = \colorGraphFilter, line width = \mmyArrowWidth] (\pathid\nodeid.north);}

\end{tikzpicture}} 
\caption{Pruned tree $\mathcal{T}$ without signal perturbation}\label{fig:orig}
\end{subfigure}
\begin{subfigure}{0.3\textwidth}
\centering{\scriptsize \begin{tikzpicture}[x = 1*\unit, y = 1*\unit,scale=\scaletrees, every node/.style={scale=1}, line/.style={
      draw,scale=1,
      -latex',
      shorten >=1pt
    },>=latex]
    
	\node at (0, 0) (o) {};

	\path (o) 
		      node (0) [fill = \colorNode,mnode, draw = \colorEdge] {};
    \foreach \dx\dy\nodeid in {{-2/-1/1},{-1/-1/2},{1/-1/4}}
         {\path (0) ++ (\dx*0.8*\mdistNodesLayerOne,\dy*\mdistBetweenLayers)
              node (\nodeid)[fill = \colorNode,mnode, draw = \colorEdge] 
              {};
        \path (0.south) 
       edge [draw = \colorGraphFilter, line width = \mmyArrowWidth] 
       (\nodeid.north east);}

    \foreach \pathid\dx\dy\nodeid in {{1/-2/-1/1},{1/-1/-1/2},
    {2/-2/-1/1},{2/-1/-1/2},
    {4/-2/-1/1},{4/-1/-1/2}
    }
         {\path (\pathid) ++ (\dx*0.6*\mdistNodesLayerTwo,\dy*\mdistBetweenLayers)
              node (\pathid\nodeid)[fill = \colorNode,mnode, draw = \colorEdge] 
              {};
              \path (\pathid.south) 
       edge [ draw = \colorGraphFilter, line width = \mmyArrowWidth] (\pathid\nodeid.north);}

\end{tikzpicture}} 
\caption{Pruned tree $\tilde{\mathcal{T}}$ with random signal perturbation}\label{fig:ran}
\end{subfigure}
\begin{subfigure}{0.3\textwidth}
\centering{\scriptsize \begin{tikzpicture}[x = 1*\unit, y = 1*\unit,scale=\scaletrees, every node/.style={scale=1}, line/.style={
      draw,scale=1,
      -latex',
      shorten >=1pt
    },>=latex]
    
	\node at (0, 0) (o) {};

	\path (o) 
		      node (0) [fill = \colorNode,mnode, draw = \colorEdge] {};
    \foreach \dx\dy\nodeid in {{-2/-1/1},{-1/-1/2},{0/-1/3},{1/-1/4}
    }
         {\path (0) ++ (\dx*0.8*\mdistNodesLayerOne,\dy*\mdistBetweenLayers)
              node (\nodeid)[fill = \colorNode,mnode, draw = \colorEdge] 
              {};
        \path (0.south) 
       edge [draw = \colorGraphFilter, line width = \mmyArrowWidth] 
       (\nodeid.north east);}

    \foreach \pathid\dx\dy\nodeid in {{1/-2/-1/1},{1/-1/-1/2},
    {2/-2/-1/1},{2/-1/-1/2},{2/0/-1/3},
    {3/-2/-1/1},{3/-1/-1/2},
    {4/-2/-1/1},{4/-1/-1/2}
    }
         {\path (\pathid) ++ (\dx*0.6*\mdistNodesLayerTwo,\dy*\mdistBetweenLayers)
              node (\pathid\nodeid)[fill = \colorNode,mnode, draw = \colorEdge] 
              {};
              \path (\pathid.south) 
       edge [ draw = \colorGraphFilter, line width = \mmyArrowWidth] (\pathid\nodeid.north);}

\end{tikzpicture}} 
\caption{Pruned tree $\tilde{\mathcal{T}}$ with localized signal perturbation}\label{fig:loc}
\end{subfigure}
    \caption{The pGST applied to the Cora dataset with and without signal perturbation, corroborating the result of Corollary 1.}
     \label{fig:pGSTexnoise}
\vspace{0mm}
\end{figure*}
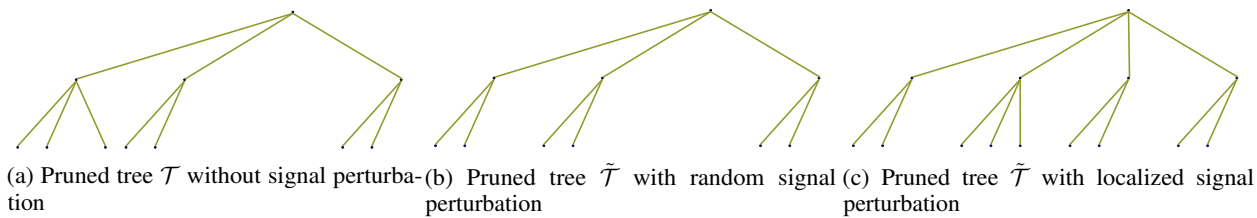

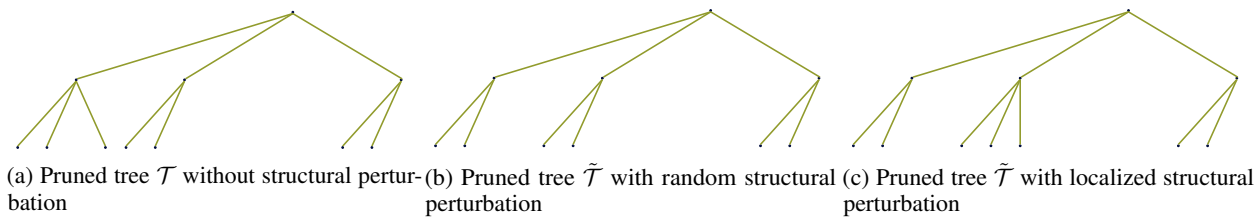
\begin{figure*}
\begin{subfigure}{0.3\textwidth}
\centering{\scriptsize \begin{tikzpicture}[x = 1*\unit, y = 1*\unit,scale=\scaletrees, every node/.style={scale=1}, line/.style={
      draw,scale=1,
      -latex',
      shorten >=1pt
    },>=latex]
    
	\node at (0, 0) (o) {};

	\path (o) 
		      node (0) [fill = \colorNode,mnode, draw = \colorEdge] {};
    \foreach \dx\dy\nodeid in {{-2/-1/1},{-1/-1/2},{1/-1/4}}
         {\path (0) ++ (\dx*0.8*\mdistNodesLayerOne,\dy*\mdistBetweenLayers)
              node (\nodeid)[fill = \colorNode,mnode, draw = \colorEdge] 
              {};
        \path (0.south) 
       edge [draw = \colorGraphFilter, line width = \mmyArrowWidth] 
       (\nodeid.north east);}

    \foreach \pathid\dx\dy\nodeid in {{1/-2/-1/1},{1/-1/-1/2},{1/1/-1/4},
    {2/-2/-1/1},{2/-1/-1/2},
    {4/-2/-1/1},{4/-1/-1/2}
    }
         {\path (\pathid) ++ (\dx*0.6*\mdistNodesLayerTwo,\dy*\mdistBetweenLayers)
              node (\pathid\nodeid)[fill = \colorNode,mnode, draw = \colorEdge] 
              {};
              \path (\pathid.south) 
       edge [ draw = \colorGraphFilter, line width = \mmyArrowWidth] (\pathid\nodeid.north);}

\end{tikzpicture}} 
\caption{Pruned tree $\mathcal{T}$ without structural perturbation}\label{fig:origstr}
\end{subfigure}
\begin{subfigure}{0.3\textwidth}
\centering{\scriptsize \begin{tikzpicture}[x = 1*\unit, y = 1*\unit,scale=\scaletrees, every node/.style={scale=1}, line/.style={
      draw,scale=1,
      -latex',
      shorten >=1pt
    },>=latex]
    
	\node at (0, 0) (o) {};

	\path (o) 
		      node (0) [fill = \colorNode,mnode, draw = \colorEdge] {};
    \foreach \dx\dy\nodeid in {{-2/-1/1},{-1/-1/2},{1/-1/4}}
         {\path (0) ++ (\dx*0.8*\mdistNodesLayerOne,\dy*\mdistBetweenLayers)
              node (\nodeid)[fill = \colorNode,mnode, draw = \colorEdge] 
              {};
        \path (0.south) 
       edge [draw = \colorGraphFilter, line width = \mmyArrowWidth] 
       (\nodeid.north east);}

    \foreach \pathid\dx\dy\nodeid in {{1/-2/-1/1},{1/-1/-1/2},
    {2/-2/-1/1},{2/-1/-1/2},
    {4/-2/-1/1},{4/-1/-1/2}
    }
         {\path (\pathid) ++ (\dx*0.6*\mdistNodesLayerTwo,\dy*\mdistBetweenLayers)
              node (\pathid\nodeid)[fill = \colorNode,mnode, draw = \colorEdge] 
              {};
              \path (\pathid.south) 
       edge [ draw = \colorGraphFilter, line width = \mmyArrowWidth] (\pathid\nodeid.north);}

\end{tikzpicture}} 
\caption{Pruned tree $\tilde{\mathcal{T}}$ with random structural perturbation}\label{fig:ranstr}
\end{subfigure}
\begin{subfigure}{0.3\textwidth}
\centering{\scriptsize \begin{tikzpicture}[x = 1*\unit, y = 1*\unit,scale=\scaletrees, every node/.style={scale=1}, line/.style={
      draw,scale=1,
      -latex',
      shorten >=1pt
    },>=latex]
    
	\node at (0, 0) (o) {};

	\path (o) 
		      node (0) [fill = \colorNode,mnode, draw = \colorEdge] {};
    \foreach \dx\dy\nodeid in {{-2/-1/1},{-1/-1/2},{1/-1/4}}
         {\path (0) ++ (\dx*0.8*\mdistNodesLayerOne,\dy*\mdistBetweenLayers)
              node (\nodeid)[fill = \colorNode,mnode, draw = \colorEdge] 
              {};
        \path (0.south) 
       edge [draw = \colorGraphFilter, line width = \mmyArrowWidth] 
       (\nodeid.north east);}

    \foreach \pathid\dx\dy\nodeid in {{1/-2/-1/1},{1/-1/-1/2},
    {2/-2/-1/1},{2/-1/-1/2},{2/0/-1/3},
    {4/-2/-1/1},{4/-1/-1/2}
    }
         {\path (\pathid) ++ (\dx*0.6*\mdistNodesLayerTwo,\dy*\mdistBetweenLayers)
              node (\pathid\nodeid)[fill = \colorNode,mnode, draw = \colorEdge] 
              {};
              \path (\pathid.south) 
       edge [ draw = \colorGraphFilter, line width = \mmyArrowWidth] (\pathid\nodeid.north);}

\end{tikzpicture}} 
\caption{Pruned tree $\tilde{\mathcal{T}}$ with localized structural perturbation}\label{fig:locstr}
\end{subfigure}
    \caption{{The pGST applied to the Cora dataset with and without structural perturbations.}
    }
     \label{fig:pGSTexnoise_structure}
\vspace{0mm}
\end{figure*}

\subsubsection{Scattering patterns insights for GCN design}
Finally, we empirically explore the connections between PGST and GCN design, in response to RQ5. We validate the scattering patterns in Fig. 1 for the Protein datasets. The validation employs DiffPoll that is a state-of-the-art GCN; see~\cite{ying2018hierarchical}.  

Fig.~\ref{fig:pGSTexap} depicts the pruned scattering patterns for the three protein datasets (Enzyme, Protein and DD). Most connections after $l=5$ are pruned. This suggests that for protein data $L=5$ graph convolution layers capture most information.

\begin{figure*}
\begin{subfigure}{0.3\textwidth}
\centering{\scriptsize \begin{tikzpicture}[x = 1*\unit, y = 1*\unit,scale=0.3, every node/.style={scale=1}, line/.style={
      draw,scale=1,
      -latex',
      shorten >=1pt
    },>=latex]
    
	\node at (0, 0) (o) {};

	\path (o) 
		      node (0) [fill = \colorNode,mnode, draw = \colorEdge] {};
    \foreach \dx\dy\nodeid in {{-1/-1/1},{0/-1/2},{1/-1/3}}
         {\path (0) ++ (\dx*\mdistNodesLayerOne,\dy*\mdistBetweenLayers)
              node (\nodeid)[fill = \colorNode,mnode, draw = \colorEdge] 
              {};
        \path (0.south) 
       edge [draw = \colorGraphFilter, line width = \mmyArrowWidth] 
       (\nodeid.north east);}

    \foreach \pathid\dx\dy\nodeid in {{1/-1/-1/1},{1/0/-1/2},{1/1/-1/3},
    {2/-1/-1/1},{2/0/-1/2},{2/1/-1/3},{3/-1/-1/1},{3/0/-1/2},{3/1/-1/3}}
         {\path (\pathid) ++ (\dx*\mdistNodesLayerTwo,\dy*\mdistBetweenLayers)
              node (\pathid\nodeid)[fill = \colorNode,mnode, draw = \colorEdge] 
              {};
              \path (\pathid.south) 
       edge [ draw = \colorGraphFilter, line width = \mmyArrowWidth] (\pathid\nodeid.north);}

    \foreach \pathid\dx\dy\nodeid in {{11/-1/-1/1},{11/0/-1/2},{11/1/-1/3},{13/-1/-1/1},{13/0/-1/2},{13/1/-1/3},{21/-1/-1/1},{21/0/-1/2},{21/1/-1/3},
    {22/-1/-1/1},{22/0/-1/2},{22/1/-1/3},{23/-1/-1/1},{23/0/-1/2},{23/1/-1/3},{31/-1/-1/1},{31/0/-1/2},{31/1/-1/3},
    {32/-1/-1/1},{32/0/-1/2},{32/1/-1/3},{33/-1/-1/1},{33/0/-1/2},{33/1/-1/3}}
         {\path (\pathid) ++ (\dx*\mdistNodesLayerThree,\dy*\mdistBetweenLayers)
              node (\pathid\nodeid)[fill = \colorNode,mnode, draw = \colorEdge] 
              {};
               \path (\pathid.south) 
              edge [ draw = \colorGraphFilter, line width = \mmyArrowWidth] (\pathid\nodeid.north);}
  
	\foreach \pathid\dx\dy\nodeid in {{111/-1/-1/1},{111/0/-1/2},{111/1/-1/3},
	{131/-1/-1/1},{131/0/-1/2},{131/1/-1/3},{211/-1/-1/1},{211/0/-1/2},{211/1/-1/3},{213/-1/-1/1},{213/0/-1/2},{213/1/-1/3},
    {222/-1/-1/1},{222/0/-1/2},{222/1/-1/3},{223/-1/-1/1},{223/0/-1/2},{223/1/-1/3},
    {311/-1/-1/1},{311/0/-1/2},{311/1/-1/3},{313/-1/-1/1},{313/0/-1/2},{313/1/-1/3},{321/-1/-1/1},{321/0/-1/2},{321/1/-1/3},{322/-1/-1/1},{322/0/-1/2},{322/1/-1/3},{323/-1/-1/1},{323/0/-1/2},{323/1/-1/3},{331/-1/-1/1},{331/0/-1/2},{331/1/-1/3},{332/-1/-1/1},{332/0/-1/2},{332/1/-1/3},{333/-1/-1/1},{333/0/-1/2},{333/1/-1/3}}
         {\path (\pathid) ++ (\dx*\mdistNodesLayerFour,\dy*\mdistBetweenLayers)
              node (\pathid\nodeid)[fill = \colorNode,mnode, draw = \colorEdge] 
              {};
               \path (\pathid.south) 
              edge [ draw = \colorGraphFilter, line width = \mmyArrowWidth] (\pathid\nodeid.north);}
    \foreach \pathid\dx\dy\nodeid in {{1111/-1/-1/1},{1111/0/-1/2},{1111/1/-1/3},
    {2221/-1/-1/1},{2221/0/-1/2},{2221/1/-1/3},
    {3321/-1/-1/1},{3321/0/-1/2},{3321/1/-1/3}}
         {\path (\pathid) ++ (\dx*\mdistNodesLayerFour,\dy*\mdistBetweenLayers)
              node (\pathid\nodeid)[fill = \colorNode,mnode, draw = \colorEdge] 
              {};
               \path (\pathid.south) 
              edge [ draw = \colorGraphFilter, line width = \mmyArrowWidth] (\pathid\nodeid.north);}
\end{tikzpicture}} 
\caption{DD}
\end{subfigure}
\begin{subfigure}{0.3\textwidth}
\centering{\scriptsize \begin{tikzpicture}[x = 1*\unit, y = 1*\unit,scale=0.3, every node/.style={scale=1}, line/.style={
      draw,scale=1,
      -latex',
      shorten >=1pt
    },>=latex]
    
	\node at (0, 0) (o) {};

	\path (o) 
		      node (0) [fill = \colorNode,mnode, draw = \colorEdge] {};
    \foreach \dx\dy\nodeid in {{-1/-1/1},{0/-1/2},{1/-1/3}}
         {\path (0) ++ (\dx*\mdistNodesLayerOne,\dy*\mdistBetweenLayers)
              node (\nodeid)[fill = \colorNode,mnode, draw = \colorEdge] 
              {};
        \path (0.south) 
       edge [draw = \colorGraphFilter, line width = \mmyArrowWidth] 
       (\nodeid.north east);}

    \foreach \pathid\dx\dy\nodeid in {{1/-1/-1/1},{1/0/-1/2},{1/1/-1/3},
    {2/-1/-1/1},{2/0/-1/2},{2/1/-1/3},{3/-1/-1/1},{3/0/-1/2},{3/1/-1/3}}
         {\path (\pathid) ++ (\dx*\mdistNodesLayerTwo,\dy*\mdistBetweenLayers)
              node (\pathid\nodeid)[fill = \colorNode,mnode, draw = \colorEdge] 
              {};
              \path (\pathid.south) 
       edge [ draw = \colorGraphFilter, line width = \mmyArrowWidth] (\pathid\nodeid.north);}

    \foreach \pathid\dx\dy\nodeid in {{11/-1/-1/1},{11/0/-1/2},{11/1/-1/3},{13/-1/-1/1},{13/0/-1/2},{13/1/-1/3},{21/-1/-1/1},{21/0/-1/2},{21/1/-1/3},
    {22/-1/-1/1},{22/0/-1/2},{22/1/-1/3},{23/-1/-1/1},{23/0/-1/2},{23/1/-1/3},{31/-1/-1/1},{31/0/-1/2},{31/1/-1/3},
    {32/-1/-1/1},{32/0/-1/2},{32/1/-1/3},{33/-1/-1/1},{33/0/-1/2},{33/1/-1/3}}
         {\path (\pathid) ++ (\dx*\mdistNodesLayerThree,\dy*\mdistBetweenLayers)
              node (\pathid\nodeid)[fill = \colorNode,mnode, draw = \colorEdge] 
              {};
               \path (\pathid.south) 
              edge [ draw = \colorGraphFilter, line width = \mmyArrowWidth] (\pathid\nodeid.north);}
  
	\foreach \pathid\dx\dy\nodeid in {{111/-1/-1/1},{111/0/-1/2},{111/1/-1/3},
	{131/-1/-1/1},{131/0/-1/2},{131/1/-1/3},{211/-1/-1/1},{211/0/-1/2},{211/1/-1/3},{213/-1/-1/1},{213/0/-1/2},{213/1/-1/3},
    {221/-1/-1/1},{221/0/-1/2},{221/1/-1/3},{222/-1/-1/1},{222/0/-1/2},{222/1/-1/3},{223/-1/-1/1},{223/0/-1/2},{223/1/-1/3},
    {231/-1/-1/1},{231/0/-1/2},{231/1/-1/3},{232/-1/-1/1},{232/0/-1/2},{232/1/-1/3},{233/-1/-1/1},{233/0/-1/2},{233/1/-1/3},
    {311/-1/-1/1},{311/0/-1/2},{311/1/-1/3},{313/-1/-1/1},{313/0/-1/2},{313/1/-1/3},{321/-1/-1/1},{321/0/-1/2},{321/1/-1/3},{322/-1/-1/1},{322/0/-1/2},{322/1/-1/3},{323/-1/-1/1},{323/0/-1/2},{323/1/-1/3},{331/-1/-1/1},{331/0/-1/2},{331/1/-1/3},{332/-1/-1/1},{332/0/-1/2},{332/1/-1/3},{333/-1/-1/1},{333/0/-1/2},{333/1/-1/3}}
         {\path (\pathid) ++ (\dx*\mdistNodesLayerFour,\dy*\mdistBetweenLayers)
              node (\pathid\nodeid)[fill = \colorNode,mnode, draw = \colorEdge] 
              {};
               \path (\pathid.south) 
              edge [ draw = \colorGraphFilter, line width = \mmyArrowWidth] (\pathid\nodeid.north);}
    \foreach \pathid\dx\dy\nodeid in {{1111/-1/-1/1},{1111/0/-1/2},{1111/1/-1/3},
	{2232/1/-1/3},{2233/-1/-1/1},{2233/0/-1/2},{3212/1/-1/3},{3213/-1/-1/1},{3213/0/-1/2}}
         {\path (\pathid) ++ (\dx*\mdistNodesLayerFour,\dy*\mdistBetweenLayers)
              node (\pathid\nodeid)[fill = \colorNode,mnode, draw = \colorEdge] 
              {};
               \path (\pathid.south) 
              edge [ draw = \colorGraphFilter, line width = \mmyArrowWidth] (\pathid\nodeid.north);}
\end{tikzpicture}} 
\caption{Protein}
\end{subfigure}
\begin{subfigure}{0.3\textwidth}
\centering{\scriptsize \begin{tikzpicture}[x = 1*\unit, y = 1*\unit,scale=0.3, every node/.style={scale=1}, line/.style={
      draw,scale=1,
      -latex',
      shorten >=1pt
    },>=latex]
    
	\node at (0, 0) (o) {};

	\path (o) 
		      node (0) [fill = \colorNode,mnode, draw = \colorEdge] {};
    \foreach \dx\dy\nodeid in {{-1/-1/1},{0/-1/2},{1/-1/3}}
         {\path (0) ++ (\dx*\mdistNodesLayerOne,\dy*\mdistBetweenLayers)
              node (\nodeid)[fill = \colorNode,mnode, draw = \colorEdge] 
              {};
        \path (0.south) 
       edge [draw = \colorGraphFilter, line width = \mmyArrowWidth] 
       (\nodeid.north east);}

    \foreach \pathid\dx\dy\nodeid in {{1/-1/-1/1},{1/0/-1/2},{1/1/-1/3},
    {2/-1/-1/1},{2/0/-1/2},{2/1/-1/3},{3/-1/-1/1},{3/0/-1/2},{3/1/-1/3}}
         {\path (\pathid) ++ (\dx*\mdistNodesLayerTwo,\dy*\mdistBetweenLayers)
              node (\pathid\nodeid)[fill = \colorNode,mnode, draw = \colorEdge] 
              {};
              \path (\pathid.south) 
       edge [ draw = \colorGraphFilter, line width = \mmyArrowWidth] (\pathid\nodeid.north);}

    \foreach \pathid\dx\dy\nodeid in {{11/-1/-1/1},{11/0/-1/2},{11/1/-1/3},{13/-1/-1/1},{13/0/-1/2},{13/1/-1/3},{21/-1/-1/1},{21/0/-1/2},{21/1/-1/3},
    {22/-1/-1/1},{22/0/-1/2},{22/1/-1/3},{23/-1/-1/1},{23/0/-1/2},{23/1/-1/3},{31/-1/-1/1},{31/0/-1/2},{31/1/-1/3},
    {32/-1/-1/1},{32/0/-1/2},{32/1/-1/3}}
         {\path (\pathid) ++ (\dx*\mdistNodesLayerThree,\dy*\mdistBetweenLayers)
              node (\pathid\nodeid)[fill = \colorNode,mnode, draw = \colorEdge] 
              {};
               \path (\pathid.south) 
              edge [ draw = \colorGraphFilter, line width = \mmyArrowWidth] (\pathid\nodeid.north);}
  
	\foreach \pathid\dx\dy\nodeid in {{111/-1/-1/1},{111/0/-1/2},{111/1/-1/3},
	{131/-1/-1/1},{131/0/-1/2},{131/1/-1/3},{211/-1/-1/1},{211/0/-1/2},{211/1/-1/3},{213/-1/-1/1},{213/0/-1/2},{213/1/-1/3},
    {223/-1/-1/1},{223/0/-1/2},{223/1/-1/3},
    {231/-1/-1/1},{231/0/-1/2},{231/1/-1/3},{232/-1/-1/1},{232/0/-1/2},{232/1/-1/3},{233/-1/-1/1},{233/0/-1/2},{233/1/-1/3},
    {311/-1/-1/1},{311/0/-1/2},{311/1/-1/3},{313/-1/-1/1},{313/0/-1/2},{313/1/-1/3},{321/-1/-1/1},{321/0/-1/2},{321/1/-1/3},{322/-1/-1/1},{322/0/-1/2},{322/1/-1/3},{323/-1/-1/1},{323/0/-1/2},{323/1/-1/3}}
         {\path (\pathid) ++ (\dx*\mdistNodesLayerFour,\dy*\mdistBetweenLayers)
              node (\pathid\nodeid)[fill = \colorNode,mnode, draw = \colorEdge] 
              {};
               \path (\pathid.south) 
              edge [ draw = \colorGraphFilter, line width = \mmyArrowWidth] (\pathid\nodeid.north);}
    \foreach \pathid\dx\dy\nodeid in {{1111/-1/-1/1},{1111/0/-1/2},{1111/1/-1/3},
    {2311/-1/-1/1},{2311/0/-1/2},{2311/1/-1/3},{232/-1/-1/1},
    {3111/-1/-1/1},{3111/0/-1/2},{3111/1/-1/3}}
         {\path (\pathid) ++ (\dx*\mdistNodesLayerFour,\dy*\mdistBetweenLayers)
              node (\pathid\nodeid)[fill = \colorNode,mnode, draw = \colorEdge] 
              {};
               \path (\pathid.south) 
              edge [ draw = \colorGraphFilter, line width = \mmyArrowWidth] (\pathid\nodeid.north);}
\end{tikzpicture}} 
\caption{Enzymes}
\end{subfigure}
    \caption{Illustration of the pGST applied to three protein datasets. Observe that after $l=5$ most features are pruned.}
     \label{fig:pGSTexap}
\vspace{0mm}
\end{figure*}
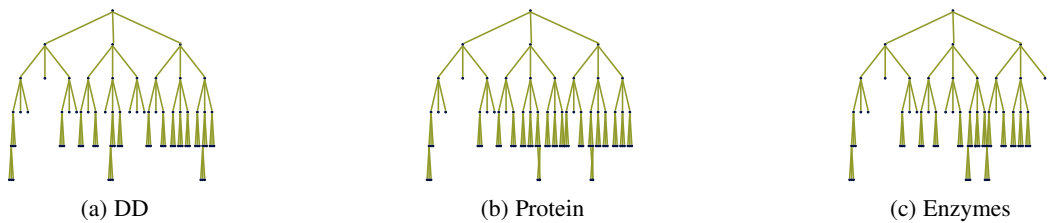

Fig.~\ref{fig:diffpool} shows the performance of DiffPool for these datasets as the number of GCN layers increases. The performance of DiffPool does not improve significantly for more than 5 GCN layers, which corroborates the insights obtained from the pGST patterns.
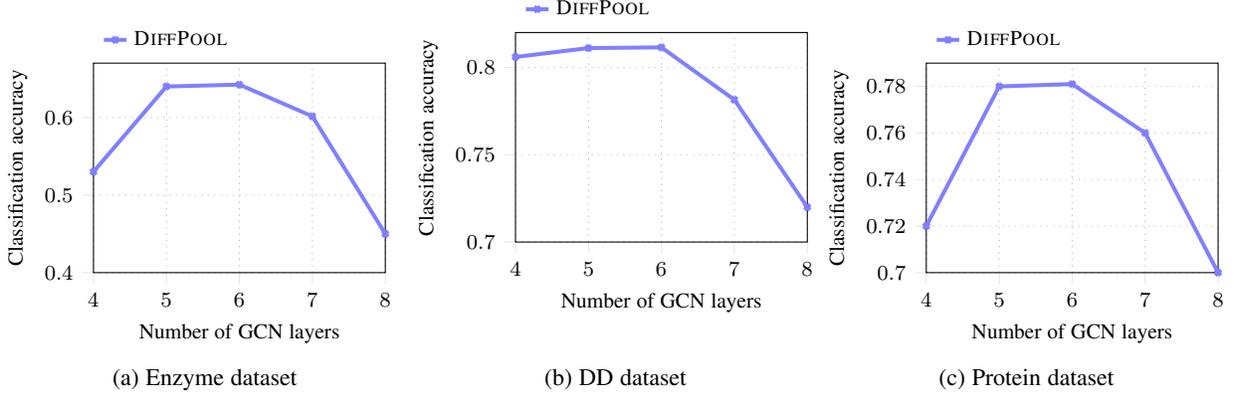
\begin{figure*}
\begin{subfigure}[b]{0.3\textwidth}
\begin{tikzpicture}

\begin{axis}[width=0.3\mywidthbb,
height=0.4\myheightbb,
at={(0\mywidth,0\myheight)},
legend entries={\textsc{DiffPool}},
legend style={draw=white!80.0!black},
tick align=outside,
tick pos=left,
x grid style={white!69.01960784313725!black},
xlabel={Number of GCN layers},ticklabel style={font=\footnotesize},
label style={font=\footnotesize},
xmin=4, xmax=8,
ylabel={Classification accuracy},
ymin=0.4, ymax=0.67,
legend columns=6,
xmajorgrids,
ymajorgrids,
grid style={dotted},
legend style={
	at={(0,1.015)}, 
	anchor=south west, legend cell align=left, align=left,
	draw=none
	, font=\footnotesize}]
\addplot [line width=\mylinewidth, DiffPool, mark=x, mark size=\markwidth, mark options={solid}]
table [row sep=\\]{
4 0.53\\
5 0.6399\\
6 0.6423\\
7 0.6015\\
8 0.45\\
};
\end{axis}
\end{tikzpicture}
\caption{Enzyme dataset}
\end{subfigure}
\begin{subfigure}[b]{0.3\textwidth}
\begin{tikzpicture}

\begin{axis}[width=0.3\mywidthbb,
height=0.4\myheightbb,
at={(0\mywidth,0\myheight)},
legend entries={\textsc{DiffPool}},
legend style={draw=white!80.0!black},
tick align=outside,
tick pos=left,
x grid style={white!69.01960784313725!black},
xlabel={Number of GCN layers},ticklabel style={font=\footnotesize},
label style={font=\footnotesize},
xmin=4, xmax=8,
ylabel={Classification accuracy},
ymin=0.70, ymax=0.82,
legend columns=6,
xmajorgrids,
ymajorgrids,
grid style={dotted},
legend style={
	at={(0,1.015)}, 
	anchor=south west, legend cell align=left, align=left,
	draw=none
	, font=\footnotesize}]
\addplot [line width=\mylinewidth, DiffPool, mark=x, mark size=\markwidth, mark options={solid}]
table [row sep=\\]{
4 0.806\\
5 0.8111\\
6 0.8115\\
7 0.7815\\
8 0.72\\
};
\end{axis}
\end{tikzpicture}
\caption{DD dataset}
\end{subfigure}
\begin{subfigure}[b]{0.3\textwidth}
\begin{tikzpicture}

\begin{axis}[width=0.3\mywidthbb,
height=0.4\myheightbb,
at={(0\mywidth,0\myheight)},
legend entries={\textsc{DiffPool}},
legend style={draw=white!80.0!black},
tick align=outside,
tick pos=left,
x grid style={white!69.01960784313725!black},
xlabel={Number of GCN layers},ticklabel style={font=\footnotesize},
label style={font=\footnotesize},
xmin=4, xmax=8,
ylabel={Classification accuracy},
ymin=0.70, ymax=0.79,
legend columns=6,
xmajorgrids,
ymajorgrids,
grid style={dotted},
legend style={
	at={(0,1.015)}, 
	anchor=south west, legend cell align=left, align=left,
	draw=none
	, font=\footnotesize}]
\addplot [line width=\mylinewidth, DiffPool, mark=x, mark size=\markwidth, mark options={solid}]
table [row sep=\\]{
4 0.72\\
5 0.78\\
6 0.7810\\
7 0.76\\
8 0.70\\
};
\end{axis}
\end{tikzpicture}
\caption{Protein dataset}\end{subfigure}
    \caption{Performance of DiffPool for different number of GCN layers.  }
    \label{fig:diffpool}
\end{figure*}

\vspace{1mm}

\section{Conclusions}
\vspace{1mm}
This paper developed a novel approach to pruning the graph scattering transform. The proposed pGST relies on a graph-spectrum-based data-adaptive criterion to prune less informative features on-the-fly, and effectively reduce the computational complexity of GSTs. Stability of pGST is established in the presence of perturbations in the input or the network structure.   Sensitivity analysis of pGST reveals that the transform is more sensitive to noise that is localized in the graph spectrum relative to noise that is uniformly spread over the spectrum, which is appealing since pGST can detect transient changes in the graph spectral domain. The novel pGST extracts efficient and stable features of graph data. Experiments demonstrated i) the performance gains of pGSTs relative to GSTs; ii) that pGST is competitive in a variety of graph classification tasks; and (iii) graph data from different domains exhibit unique pruned scattering patterns, which calls for adaptive network architectures. 
\section{Proofs}
\subsection{Proof of Theorem 1}
By its definition, the objective in \eqref{eq:criterion} can be rewritten as 
\begin{align}
    &\sum^{J}_{j=1}
\left(\sum_{n=1}^N\left(\widehat{{h}}_j(\lambda_n)^2-\tau\right)[\widehat{\mathbf{z}}_{(p)}]^2_n\right) {f_j}\nonumber\\=&\sum^{J}_{j=1}\widehat{\mathbf{z}}_{(p)}\transpose\left(\diag{(\widehat{h}_j(\bm{\lambda}))}^2-\tau\identitymat\right)\widehat{\mathbf{z}}_{(p)} 
    {f_j} \;.\nonumber\end{align}
By introducing the scalars $\alpha_j\define \widehat{\mathbf{z}}_{(p)}\transpose(\diag{(\widehat{h}_j(\bm{\lambda}))}^2-\tau\identitymat)\widehat{\mathbf{z}}_{(p)}$ for $j=1,\ldots,J$, \eqref{eq:criterion} can be rewritten as
\begin{align}
\label{eq:criterionsum}
    \max_{f_{j}}&~~~~~\sum^{J}_{j=1}\alpha_j
    {f_j}\\
    \text{s. t.}&~~~~~f_j \in \{0,1\},~~j=1,\ldots,{J}\nonumber.
\end{align}
The optimization problem in \eqref{eq:criterionsum} is nonconvex  since $f_j$ is a discrete variable. However, maximizing the sum in \eqref{eq:criterionsum} amounts to setting $f_j=1$ for the positive $\alpha_j$ over $j$. Such an approach leads to the optimal pruning assignment variables
\begin{align}
        f_{j}^* =\begin{dcases*} 1
   & if  $\alpha_j>0,$ \\[0.5ex]
0
   & if $\alpha_j<0$.
\end{dcases*},~j=1,\ldots,{J}
\end{align}
The rest of the proof focuses on rewriting $\alpha_j$ as
\begin{align}
    \alpha_j=&\widehat{\mathbf{z}}_{(p)}\transpose(\diag{(\widehat{h}_j(\bm{\lambda}))}^2-\tau\identitymat)\widehat{\mathbf{z}}_{(p)}\\=&
    \|{\diag{(\widehat{h}_j(\bm{\lambda}))}\widehat{\mathbf{z}}\scatternot{{p}}}\|^2
-\tau\|\widehat{\mathbf{z}}_{(p)}\|^2
\end{align}
Furthermore, since matrix $\mathbf{V}$ is orthogonal, it holds that 
$\|\widehat{\mathbf{z}}_{(p)}\|^2=\|\mathbf{V}\transpose\mathbf{z}_{(p)}\|^2=\|\mathbf{z}_{(p)}\|^2$ from which it follows that  

\begin{align}
\label{eq:gftdecomp}
    \|{\diag{(\widehat{h}_j(\bm{\lambda}))}\widehat{\mathbf{z}}\scatternot{{p}}}\|^2=&\|{{h}_j(\mathbf{S})\mathbf{z}\scatternot{{p}}}\|^2\\
    =&\|\nonlinearity{{h}_j(\mathbf{S})\mathbf{z}\scatternot{{p}}}\|^2\nonumber\\
    =&\|\mathbf{z}_{(p, j)}\|^2
\end{align}
where the second line follows because $\sigma(\cdot)$ is applied elementwise, and does not change the norm.

\subsection{Proof of Lemma 1}
By the definition in \eqref{eq:defscatvect}, it holds that 
\begin{align}
    \|\bm{\Phi}(\mathbf{x})-\bm\Phi (\tilde{\mathbf{x}})\|^2
    =\sum_{\ell=0}^{L} \sum_{p^{(\ell)}\in\mathcal{P}^{(\ell)}}
    | \phi_{(p^{(\ell)})}-
    \psummarizedscatterfeat_{(p^{(\ell)})}|^2
    \label{eq:scattransdefapl}
\end{align} 
which suggests bounding each summand in \eqref{eq:scattransdefapl}  as
\begin{align}
    | \phi_{(p^{(\ell)})}-\psummarizedscatterfeat_{(p^{(\ell)})}|=&| U(\mathbf{z}_{(p^{(\ell)})}- U(\tilde{\mathbf{z}}_{(p^{(\ell)})}    )|\\
    \le&\| U\| \|\mathbf{z}_{(p^{(\ell)})}-\tilde{\mathbf{z}}_{(p^{(\ell)})}\|\label{eq:tr1}\end{align}
where \eqref{eq:tr1} follows since the norm is a sub-multiplicative operator. Next, we will show the recursive bound
\begin{align}
    &\|\mathbf{z}_{(p^{(\ell)})}-\tilde{\mathbf{z}}_{(p^{(\ell)})}\|\nonumber\\=& \|\nonlinearity{
    h_{j^{(\ell)}}(\mathbf{S})
    \mathbf{z}_{(p\layernot{\ell-1}})}
    -\nonlinearity{
    h_{j^{(\ell)}}(\mathbf{S})
    \tilde{\mathbf{z}}_{(p\layernot{\ell-1})}}\|\\
    \le& \|\nonlinearity{}\|\| 
    h_{j^{(\ell)}}(\mathbf{S})
    \mathbf{z}_{(p\layernot{\ell-1})}
    -
    h_{j^{(\ell)}}(\mathbf{S})
    \tilde{\mathbf{z}}_{(p\layernot{\ell-1})}\|\label{eq:tr2}\\
    \le& \| 
    h_{j^{(\ell)}}(\mathbf{S})
    \mathbf{z}_{(p\layernot{\ell-1})}
    -
    h_{j^{(\ell)}}(\mathbf{S})
    \tilde{\mathbf{z}}_{(p\layernot{\ell-1})}\|\label{eq:nonexp}\\
    \le&\|h_{j^{(\ell)}}(\mathbf{S})\|\|
    \mathbf{z}_{(p\layernot{\ell-1})}
    -
    \tilde{\mathbf{z}}_{(p\layernot{\ell-1})}\|\label{eq:tr3}
\end{align}
where  \eqref{eq:tr2}, \eqref{eq:tr3} hold because the norm is a sub-multiplicative operator, and \eqref{eq:nonexp} follows since the nonlinearity is nonexpansive, i.e. $ \|\nonlinearity{}\|<1$.
Hence, by applying \eqref{eq:tr3} $\ell-1$ times, the following condition holds
\begin{align}
\label{eq:tr4}
    \|\mathbf{z}_{(p^{(\ell)})}-\tilde{\mathbf{z}}_{(p^{(\ell)})}\|\le\|h_{j^{(\ell)}}(\mathbf{S})\|\|h_{j\layernot{\ell-1}}(\mathbf{S})\|\cdots\|h_{j^{(1)}}(\mathbf{S})\| \|
    \mathbf{x}
    -
    \tilde{\mathbf{x}}\|
\end{align}
and by further applying the frame bound and \eqref{eq:simplvecpert}, we deduce that
\begin{align}
     \|\mathbf{z}_{(p^{(\ell)})}-\tilde{\mathbf{z}}_{(p^{(\ell)})}\|\le\framebound^\ell\|{\bm{\delta}}\| \label{eq:reqrel}
\end{align}
Combining \eqref{eq:tr1}, \eqref{eq:reqrel} and the average operator property $\| U\|=1$ it holds that
\begin{align}
    | \phi_{(p^{(\ell)})}-\psummarizedscatterfeat_{(p^{(\ell)})}|\le\framebound^\ell\|{\bm{\delta}}\| \label{eq:pboundres}
\end{align}

By applying the bound \eqref{eq:pboundres} for all entries in the right hand side of \eqref{eq:scattransdefapl} it follows that 
\begin{align}
\|\bm{\Phi}(\mathbf{x})-\bm\Phi (\tilde{\mathbf{x}})\|^2
\le \sum_{\ell=0}^{L} \sum_{p^{(\ell)}\in\mathcal{P}^{(\ell)}}
\framebound^{2l}\|{\bm{\delta}}\|^2\label{eq:prefincond}
\end{align}
By factoring out $\|{\bm{\delta}}\|$ and observing that the sum in the right side of \eqref{eq:prefincond} does not depend on the path index $p$ it follows that
\begin{align}
    \|\bm{\Phi}(\mathbf{x})-\bm\Phi (\tilde{\mathbf{x}})\|^2\le
    \left(
    \sum_{\ell=0}^{L}|\mathcal{P}^{(\ell)}|\framebound^{2\ell}
    \right)\|{\bm{\delta}}\|^2
\end{align}
Finally, since the cardinality of the paths at $\ell$ is $|\mathcal{P}^{(\ell)}|=
{J}^\ell$ and $\sum_{\ell=0}^{L}(\framebound^{2}{J})^\ell=
\left((\framebound^{2}{J})^L\right)/
\left(\framebound^{2}{J}-1\right)$ it holds 
\begin{align}
        \|\bm{\Phi}(\mathbf{x})-\bm\Phi (\tilde{\mathbf{x}})\|\le
   \sqrt{\frac{(\framebound^{2}{J})^L}
{\framebound^{2}{J}-1}}\|{\bm{\delta}}\|
\end{align}

\subsection{Proof of Lemma 2}
We will prove the case for $\ell=0$, where $\mathbf{z}_{p^{(0)}}=\mathbf{x}$, since the same proof holds for any $\ell$. First, we adapt \eqref{eq:thassump} to the following
\begin{align}
\big|\|{h}_j(\mathbf{S})\mathbf{x}\|^2
-\tau\|\mathbf{x}\|^2\big|> \|{h}_j(\mathbf{S}){\bm{\delta}}\|^2+
    \tau\big|
    \|\mathbf{x}\|^2-\|\tilde{\mathbf{x}}\|^2\big|. \label{eq:assump1}
\end{align}
The proof will examine two cases and will follow by contradiction. For the first case, consider that branch $j$ is pruned in 
$\bm{\Psi}({\mathbf{x}})$ and not pruned in $\bm{\Psi}({\tilde{\mathbf{x}}})$, i.e. $(j)\notin\mathcal{T}$ and $(j)\in\tilde{\mathcal{T}}$. By applying \eqref{eq:pwcriterion} for $\mathbf{z}_{(j)}=\nonlinearity{{h}_j(\mathbf{S})\mathbf{x}}
$ there exists $C\ge0$ such that
\begin{align}
 \frac{\|{h}_j(\mathbf{S})\mathbf{x}\|^2}
{\|\mathbf{x}\|^2}\le& \tau -C\label{eq:c0}\\
\|{h}_j(\mathbf{S})\mathbf{x}\|^2
\le& \tau\|\mathbf{x}\|^2 -C\|\mathbf{x}\|^2
\label{eq:c1}
\end{align}
Furthermore, from \eqref{eq:pwcriterion} it holds for $\tilde{\mathbf{z}}_{(j)}=\nonlinearity{{h}_j(\mathbf{S})\tilde{\mathbf{x}}}
$ that
\begin{align}
 \frac{\|{h}_j(\mathbf{S})\tilde{\mathbf{x}}\|^2}
{\|\tilde{\mathbf{x}}\|^2}>\tau\label{eq:c2}
\end{align}
By applying \eqref{eq:simplvecpert} to \eqref{eq:c2}, and using the triangular inequality it follows that
\begin{align}
  \|{h}_j(\mathbf{S})\mathbf{x}\|^2+
  \|{h}_j(\mathbf{S}){\bm{\delta}}\|^2\ge\tau\|\tilde{\mathbf{x}}\|^2
\end{align}
Next, by applying \eqref{eq:c1} it holds that
\begin{align}
 \tau\|\mathbf{x}\|^2 -C\|\mathbf{x}\|^2+
  \|{h}_j(\mathbf{S}){\bm{\delta}}\|^2\ge&\tau\|\tilde{\mathbf{x}}\|^2\\
 \tau(\|\mathbf{x}\|^2-\|\tilde{\mathbf{x}}\|^2)+
  \|{h}_j(\mathbf{S}){\bm{\delta}}\|^2\ge& C\|\mathbf{x}\|^2.
\label{eq:pfeatres}
\end{align}
Next, by utilizing \eqref{eq:assump1} and the absolute value property $|a|\ge a$ to upper-bound the left side of \eqref{eq:pfeatres} it follows that
\begin{align}
\|{h}_j(\mathbf{S})\mathbf{x}\|^2
-\tau\|\mathbf{x}\|^2>& C\|\mathbf{x}\|^2.
\label{eq:pfeatres1}
\end{align}
Finally, by applying \eqref{eq:c1} the following is obtained
\begin{align}
0> 2C\|\mathbf{x}\|^2
\label{eq:pfeatres2}
\end{align}
which implies that $C<0$. However, this contradicts \eqref{eq:c0} since $C\ge 0$.
Following a similar argument we can complete the proof for the other case.

\subsection{Proof of Lemma 3}
The proof of Lemma 3 requires the following result. 

\begin{lemma}\label{th:stab1}
Consider the shift matrix $\mathbf{S}$ and the perturbed matrix $\tilde{\mathbf{S}}$, such that  $d(\mathbf{S},\tilde{\mathbf{S}})\le\mathscr{E}/2$. Further, for  $\bm{\Delta}\in \mathcal{D}$ with eigendecomposition $\bm{\Delta}=\mathbf{U}\diag{(\mathbf{d})}\mathbf{U}\transpose$ it holds that  $\|\bm{\Delta}/d_{\textrm{max}}-\mathbf{I}\|\le\mathscr{E}$, where $d_{\textrm{max}}$ is the eigenvalue of $\bm{\Delta}$ with maximum absolute value. At layer $\ell$ consider the scattering feature indexed by $p^{(\ell)}$ of the original transform ${\bm{\Psi}}({\mathbf{x}})$ as ${\mathbf{z}}_{(p^{(\ell)})}$ and of the perturbed transform $\tilde{\bm{\Psi}}({\mathbf{x}})$ as $\tilde{\mathbf{z}}_{(p^{(\ell)})}$. Suppose also that graph filter bank forms a frame with bound $B$, and $h(\lambda)$ satisfies the integral Lipschitz constraint $|\lambda h'(\lambda)|\le C_0$. It then holds that
\begin{align}
\label{eq:lem4cond}
    \|\tilde{\mathbf{z}}_{(p^{(\ell)}}-\mathbf{z}_{(p^{(\ell)})}\|\le \ell\mathscr{E} C_0 B^{\ell-1}\|\mathbf{x}\| .
\vspace{-4mm}
\end{align}
\end{lemma}
\begin{proof}
First, we add and substract on \eqref{eq:lem4cond} the term $\sigma(h_{j^{(\ell)}}(\mathbf{S})\sigma(h_{j\layernot{\ell-1}}(\tilde{\mathbf{S}})\cdots\sigma(h_{j^{(1)}}(\tilde{\mathbf{S}})x)\cdots)$ 
\begin{align}
    &\|\tilde{\mathbf{z}}_{(p^{(\ell)}}-\mathbf{z}_{(p^{(\ell)})}\|\\
     &\le\|\sigma(h_{j^{(\ell)}}(\mathbf{S})\sigma(h_{j\layernot{\ell-1}}(\tilde{\mathbf{S}})\cdots\sigma(h_{j^{(1)}}(\tilde{\mathbf{S}})\mathbf{x})\cdots))\nonumber\\&-\sigma(h_{j^{(\ell)}}(\mathbf{S})\sigma(h_{j\layernot{\ell-1}}({\mathbf{S}})\cdots\sigma(h_{j^{(1)}}({\mathbf{S}})\mathbf{x})\cdots))
    \nonumber\\
    &-\sigma(h_{j^{(\ell)}}(\mathbf{S})\sigma(h_{j\layernot{\ell-1}}(\tilde{\mathbf{S}})\cdots\sigma(h_{j^{(1)}}(\tilde{\mathbf{S}})\mathbf{x})\cdots))\nonumber\\&+\sigma(h_{j^{(\ell)}}(\tilde{\mathbf{S}})\sigma(h_{j\layernot{\ell-1}}(\tilde{\mathbf{S}})\cdots\sigma(h_{j^{(1)}}(\tilde{\mathbf{S}})\mathbf{x})\cdots))\|\\
    &\le\|\sigma(h_{j^{(\ell)}}(\mathbf{S})\sigma(h_{j\layernot{\ell-1}}(\tilde{\mathbf{S}})\cdots\sigma(h_{j^{(1)}}(\tilde{\mathbf{S}})\mathbf{x})\cdots))\nonumber\\&-\sigma(h_{j^{(\ell)}}(\mathbf{S})\sigma(h_{j\layernot{\ell-1}}({\mathbf{S}})\cdots\sigma(h_{j^{(1)}}({\mathbf{S}})\mathbf{x})\cdots))\| 
    \nonumber\\
    &+\|\sigma(h_{j^{(\ell)}}(\mathbf{S})\sigma(h_{j\layernot{\ell-1}}(\tilde{\mathbf{S}})\cdots\sigma(h_{j^{(1)}}(\tilde{\mathbf{S}})\mathbf{x})\cdots))\nonumber\\&-\sigma(h_{j^{(\ell)}}(\tilde{\mathbf{S}})\sigma(h_{j\layernot{\ell-1}}(\tilde{\mathbf{S}})\cdots\sigma(h_{j^{(1)}}(\tilde{\mathbf{S}})\mathbf{x})\cdots))\|\label{eq:ft2}
\end{align}
By a similar argument as in~\eqref{eq:tr4}, it holds for the first summand in \eqref{eq:ft2} that 
\begin{align}
    \|&\sigma(h_{j^{(\ell)}}(\mathbf{S})\sigma(h_{j\layernot{\ell-1}}(\tilde{\mathbf{S}})\cdots\sigma(h_{j^{(1)}}(\tilde{\mathbf{S}}) \mathbf{x})\nonumber\\&-\sigma(h_{j^{(\ell)}}(\mathbf{S})\sigma(h_{j\layernot{\ell-1}}({\mathbf{S}})\cdots\sigma(h_{j^{(1)}}({\mathbf{S}})\mathbf{x})\|\nonumber\\
    &\le B \|\tilde{\mathbf{z}}_{(p^{(\ell-1)}}-\mathbf{z}_{(p^{(\ell-1)})}\|\label{eq:recc1}
\end{align}
For the second summand in~\eqref{eq:ft2}, we consider the following
\begin{align}
\|\mathbf{z}_{(p^{(\ell)})}\|&\le\|h_{j^{(\ell)}}(\mathbf{S})\|\|h_{j\layernot{\ell-1}}(\mathbf{S})\|\cdots\|h_{j^{(1)}}(\mathbf{S})\| \|
    \mathbf{x}\|    \nonumber\\
    &\le B^{\ell}\|\mathbf{x}\|\label{eq:totcond}
\end{align}
which can be derived following~\eqref{eq:tr4}. Also by utilizing the result of Proposition 2 in \cite{gama2019stability} it can be shown that 
\begin{align}
    \|h_{j}({\mathbf{S}})-h_{j}(\tilde{\mathbf{S}})\|\le \mathscr{E} C_0.\label{eq:loccond}
\end{align}
where $C_0$ is the integral Lipschitz constant $|\lambda h'(\lambda)|\le C_0$.
By combining \eqref{eq:recc1}-\eqref{eq:loccond} we arrive to the following recursive condition
\begin{align}
    \|\tilde{\mathbf{z}}_{(p^{(\ell)}}-\mathbf{z}_{(p^{(\ell)})}\|\le\|\tilde{\mathbf{z}}_{(p^{(\ell-1)}}-\mathbf{z}_{(p^{(\ell-1)}}\|+(\ell-1)\mathscr{E} C_0 B^{\ell-1}\|\mathbf{x}\|\nonumber
\end{align}
which can be solved to arrive at \eqref{eq:lem4cond}.
\end{proof}

Now that Lemma 4 is established we proceed with the proof for Lemma 3 follows. The proof will examine two cases and will follow by contradiction. For the first case, consider that branch $j$ is pruned in 
$\bm{\Psi}({\mathbf{x}})$ and not pruned in $\tilde{\bm{\Psi}}({\mathbf{x}})$, i.e. $(j)\notin\mathcal{T}$ and $(j)\in\tilde{\mathcal{T}}$. By applying \eqref{eq:pwcriterion} for $\mathbf{z}_{(p,j)}=\nonlinearity{{h}_j(\mathbf{S})\mathbf{z}_{(p)}}
$ there exists $C\ge0$ such that
\begin{align}
 \frac{\|{h}_j(\mathbf{S})\mathbf{z}_{(p)}\|^2}
{\|\mathbf{z}_{(p)}\|^2}\le& \tau -C\label{eq:c0alt}\\
\|{h}_j(\mathbf{S})\mathbf{z}_{(p)}\|^2
\le& \tau\|\mathbf{z}_{(p)}\|^2 -C\|\mathbf{z}_{(p)}\|^2
\label{eq:c1alt}
\end{align}
Furthermore, from \eqref{eq:pwcriterion} it holds for $\tilde{\mathbf{z}}_{(p,j)}=\nonlinearity{{h}_j(\mathbf{S})\tilde{\mathbf{z}}_{(p)}}
$ that
\begin{align}
 \frac{\|{h}_j(\tilde{\mathbf{S}})\tilde{\mathbf{z}}_{(p)}\|^2}
{\|\tilde{\mathbf{z}}_{(p)}\|^2}>\tau\label{eq:c2alt}
\end{align}
By using the triangle inequality to \eqref{eq:c2alt}, we obtain 
\begin{align}
  \|{h}_j(\mathbf{S})\mathbf{z}_{(p)}\|^2+
  \|{h}_j(\tilde{\mathbf{S}})\tilde{\mathbf{z}}_{(p)}-{h}_j(\mathbf{S})\mathbf{z}_{(p)}\|^2\ge\tau\|\tilde{\mathbf{z}}_{(p)}\|^2
\end{align}
and upon applying \eqref{eq:c1alt}, we arrive at 
\begin{align}
 \tau(\|\mathbf{z}_{(p)}\|^2-\|\tilde{\mathbf{z}}_{(p)}\|^2)+
  \|{h}_j(\tilde{\mathbf{S}})\tilde{\mathbf{z}}_{(p)}-{h}_j(\mathbf{S})\mathbf{z}_{(p)}\|^2\ge& C\|\mathbf{z}_{(p)}\|^2.
\label{eq:pfeatresalt}
\end{align}
Next, we leverage Lemma 4 and the absolute value property $|a|\ge a$ to upper-bound the left side of \eqref{eq:pfeatresalt} as
\begin{align}
\tau\|\bm{\delta}_{(p)}\|^2+(\ell\mathscr{E} C_0 B^{\ell-1}\|\mathbf{x}\|)^2>C\|\mathbf{z}_{(p)}\|^2.
\label{eq:pfeatres1alt}
\end{align}
By applying \eqref{eq:thassump1} in Lemma 3, we deduce that 
\begin{align}
\|{h}_j(\mathbf{S})\mathbf{z}_{p}\|^2
-\tau\|\mathbf{z}_{p}\|^2> C\|\mathbf{z}_{(p)}\|^2.
\label{eq:pfeatres0alt}
\end{align}
Finally, by applying \eqref{eq:c1alt}, we find 
\begin{align}
0> 2C\|\mathbf{z}_{(p)}\|^2
\label{eq:pfeatres2alt}
\end{align}
which implies that $C<0$. However, this contradicts \eqref{eq:c0alt} since $C\ge 0$.
Following a symmetric argument, we can complete the proof for the other case too.

\bibliographystyle{IEEEtran}
\bibliography{my_bibliography}
\end{document}